\RequirePackage{algorithm} 
\documentclass[final,USenglish,hyperfootnotes=false,hypertexnames=false]{dmtcs-episciences}
\usepackage{xcolor}
\colorlet{darkgreen}{green!50!black}
\colorlet{dg}{darkgreen}
\colorlet{medgray}{gray!75}
\colorlet{lightgray}{gray!30}
\colorlet{llightgray}{lightgray!30}
\colorlet{bagcol}{gray!70}
\colorlet{pastcol}{gray!20}
\definecolor{linkcol}{rgb}{0,0,0.4}
\definecolor{citecol}{rgb}{0.5,0,0}

\usepackage{tikz}
\usetikzlibrary{arrows.meta,shapes,positioning,calc,decorations.pathreplacing,decorations.markings,decorations.pathmorphing,patterns}
\pgfdeclarelayer{background2}
\pgfdeclarelayer{background}
\pgfdeclarelayer{foreground}
\pgfsetlayers{background2,background,main,foreground}

\tikzstyle{hidden}=[opacity=0]
\tikzstyle{fade}=[opacity=0.2]
\tikzstyle{nonsol}=[dashed]

\tikzstyle{bold}=[draw, line width=2pt]
\tikzstyle{gamma}=[draw, line width=5pt]
\tikzstyle{optional}=[dashed]
\tikzstyle{path}=[decorate, decoration={snake, amplitude=.6mm}]

\tikzstyle{small}=[inner sep=2pt]
\tikzstyle{tiny}=[inner sep=1.7pt]
\tikzstyle{textnode}=[inner sep=0pt]

\tikzstyle{triangle}=[draw, regular polygon, regular polygon sides=3]
\tikzstyle{hex}=[draw, regular polygon, regular polygon sides=6]
\tikzstyle{vertex}=[circle, draw, fill=white]
\tikzstyle{reti}=[vertex, fill=black]
\tikzstyle{leaf}=[vertex, rectangle]
\tikzstyle{leaf2}=[vertex, triangle]

\tikzstyle{smallvertex}=[vertex, small]
\tikzstyle{smallvertex0}=[vertex, star, small]
\tikzstyle{smallvertex1}=[vertex, small, diamond]
\tikzstyle{smallvertex2}=[vertex, triangle, small]
\tikzstyle{smallvertex3}=[vertex, hex, inner sep=3pt]
\tikzstyle{smallvertex4}=[smallvertex]
\tikzstyle{smallleaf}=[leaf, inner sep=3.3pt]
\tikzstyle{smallleaf2}=[leaf2, inner sep=1.7pt]
\tikzstyle{smalltriangle}=[triangle, inner sep=1.5pt]
\tikzstyle{smallreti}=[reti, small]

\tikzstyle{tinyvertex}=[vertex, tiny]

\tikzstyle{normal}=[smallvertex, fill=black]

\tikzstyle{edge}=[draw, -]
\tikzstyle{matching}=[edge, line width=3pt]
\tikzstyle{arc}=[draw, arrows={-Latex[length=6pt]}]
\tikzstyle{boldarc}=[draw, bold, arrows={-Latex[length=10pt]}]
\tikzstyle{revarc}=[draw, arrows={Latex[length=6pt]-}]
\tikzstyle{boldrevarc}=[draw, bold, arrows={Latex[length=10pt]-}]

\tikzstyle{bag}=[bagcol, line width=15pt]
\tikzstyle{past}=[draw=gray, fill=pastcol]

\newcommand{\nextnode}[5][vertex]{\node[small#1] (#2) at ($(#3)+(#4)$) {} edge[revarc, #5] (#3);}

\usepackage{etoolbox}
\usepackage{microtype}
\usepackage{enumerate}
\usepackage{amsthm}
\usepackage{amssymb}
\usepackage{mathtools}
\usepackage[anycase]{thmtools}
\usepackage{thm-restate}
\usepackage[numbers]{natbib}
\usepackage{doi}
\usepackage{xurl}
\newcommand{\hairspace}{\kern .05em}
\newcommand{\ceq}{\coloneqq}
\newcommand{\bigO}{\mathcal{O}}
\newcommand{\nothing}{\varnothing}

\newcommand{\inV}{V^-}
\newcommand{\outV}{V^+}
\newcommand{\inA}{A^-}
\newcommand{\outA}{A^+}
\newcommand{\indeg}{\deg^-}
\newcommand{\outdeg}{\deg^+}
\newcommand{\outDelta}[1]{\Delta_{#1}}
\newcommand{\id}{\mathrm{id}}

\newcommand{\YES}{\mathtt{YES}}
\newcommand{\NO}{\mathtt{NO}}
\newcommand{\TC}{\textsc{Tree Containment}}
\newcommand{\STC}{\textsc{Soft Tree Containment}}
\newcommand{\NC}{\textsc{Network Containment}}
\newcommand{\CC}{\textsc{Cluster Containment}}
\newcommand{\MC}{\textsc{Minor Containment}}
\newcommand{\GWS}{\mathit{GWS}}
\newcommand{\HWS}{\mathit{HWS}}
\newcommand{\eqcom}{\text{\,,}}
\newcommand{\eqdot}{\text{\,.}}

\newcommand{\proofitem}[1]{``#1'':}

\newcommand{\proofright}{\proofitem{$\Rightarrow$}}
\newcommand{\proofleft}{\proofitem{$\Leftarrow$}}
\newcommand{\card}[1]{\lvert#1\rvert}

\DeclareMathOperator{\getroot}{root}
\DeclareMathOperator{\GW}{GW}
\DeclareMathOperator{\sw}{sw}

\DeclareMathOperator{\str}{str}
\DeclareMathOperator{\HW}{HW}
\DeclareMathOperator{\SW}{SW}
\usepackage[rightComments=false,commentColor=black,beginComment=\quad$\triangleright$~]{algpseudocodex}
\algrenewcommand\algorithmicrequire{\textbf{Input:}}
\algrenewcommand\algorithmicensure{\textbf{Output:}}
\theoremstyle{plain}
\newtheorem{theorem}{Theorem}
\newtheorem{proposition}{Proposition}
\newtheorem{lemma}{Lemma}
\newtheorem{assumption}{Assumption}
\theoremstyle{definition}
\newtheorem{definition}{Definition}
\newtheorem{problem}{Problem}
\usepackage{xargs}
\newcommandx{\problemdefref}[7][3=Input,5=Question]{
  \begin{problem}[#7]\label{prob:#2}
  \colorbox{gray!17!white}{\textsc{#1}}\nopagebreak[4]
  \par\noindent\hangindent=\parindent\textbf{#3:} #4\nopagebreak[4]
  \par\noindent\hangindent=\parindent\textbf{#5:} #6
  \end{problem}
}
\usepackage{hyperref}
\hypersetup{pdftitle={Exploiting Low Scanwidth to Resolve Soft Polytomies},
pdfauthor={Sebastian Bruchhold and Mathias Weller},
pdfsubject={Soft Tree Containment},
pdfkeywords={phylogenetic networks, containment relations, directed width measures}}
\usepackage[capitalize,nameinlink]{cleveref}
\Crefname{assumption}{Assumption}{Assumptions}
\Crefname{figure}{Figure}{Figures}
\Crefname{problem}{Problem}{Problems}
\Crefname{section}{Section}{Sections}
\title{Exploiting Low Scanwidth to Resolve Soft Polytomies\texorpdfstring{\thanks{This paper presents a result
from the first author's master's thesis \cite{Bru24},
which was supervised by the second author.
An extended abstract of this paper appeared in the proceedings of SOFSEM'26 \cite{BW26}.}}{}}
\author{Sebastian Bruchhold\affiliationmark{1}\thanks{Funded by the Deutsche Forschungsgemeinschaft
(DFG\@, German Research Foundation) under Germany's Excellence Strategy -- The Berlin Mathematics Research Center
MATH+ (EXC-2046/1, EXC-2046/2, project ID\@: 390685689).}
\and Mathias Weller\affiliationmark{2}}
\affiliation{
Technische Universit\"at Berlin, Berlin, Germany\\
LIGM\@, CNRS\@, Univ Gustave Eiffel, F77454 Marne-la-Vall\'ee, France}
\begin{document}
\publicationdata{vol. 28:4, SOFSEM 2026}{2026}{5}{10.46298/dmtcs.17720}{2026-03-14; 2026-03-14; 2026-06-30}{2026-07-25}
\makeatletter 
\preto\@firsthead{\hbadness=10000}
\appto\@oddfoot{\hss}
\makeatother
\maketitle
\begin{abstract}
  Phylogenetic networks allow modeling reticulate evolution,
  capturing events such as hybridization and horizontal gene transfer.
  A fundamental computational problem in this context is the \textsc{Tree Containment} problem,
  which asks whether a given phylogenetic network is compatible with a given phylogenetic tree.
  However, the classical statement of the problem is not robust to poorly supported branches in biological data,
  possibly leading to false negatives.
  In an effort to address this, a relaxed version that accounts for uncertainty, called \textsc{Soft Tree Containment},
  has been introduced by Bentert, Mal\'ik, and Weller~[SWAT'18].
  We present an algorithm that solves \textsc{Soft Tree Containment} in
  $2^{\mathcal{O}(\Delta_T \cdot k \cdot \log(k))} \cdot n^{\mathcal{O}(1)}$ time,
  where $k \coloneqq \operatorname{sw}(\Gamma) + \Delta_N$,
  with $\Delta_T$ and $\Delta_N$ denoting the maximum out-degrees in the tree and the network,
  respectively, and $\operatorname{sw}(\Gamma)$ denoting the ``scanwidth''~[Berry, Scornavacca, and Weller, SOFSEM'20]
  of a given tree extension of the network,
  while $n$ is the input size.
  Our approach leverages the fact that phylogenetic networks encountered in practice often exhibit low scanwidth,
  making the problem more tractable.
\end{abstract}

\section{Introduction}

Evolutionary processes are not always represented well by tree-like models, as
hybridization, horizontal gene transfer, and other reticulate events often give rise to more complex structures~\cite{journal-CCR13,journal-TR11}.
To capture these evolutionary relationships, researchers use phylogenetic networks,
which generalize phylogenetic trees by allowing for ``reticulations'' that represent non-tree-like events
(see the monographs by Gusfield~\cite{bk-Gus14} and Huson, Rupp, and Scornavacca~\cite{bk-HRS10}).

A fundamental computational problem arising in this context is the \TC{} problem,
which asks whether a given phylogenetic network is ``compatible'' with
(``firmly displays'' (see \cref{def:fdsd})) a given phylogenetic tree,
that is, whether the network contains a subdivision of the tree (as a subdigraph respecting leaf-labels).
This question is important for reconstruction methods and in the estimation of distances between evolutionary scenarios.

While the \TC{} problem is known to be NP-hard~\cite{journal-KNT+08,journal-ISS10},
numerous special cases of practical interest have been shown to admit polynomial-time
solutions~\cite{conf-GGL+15,conf-FKP15,journal-BS16,journal-GDZ17,journal-KNT+08,journal-ISS10,arxiv-Gun17,Wel18,IJW26}.
However, the classical formulation of \TC{}
is not robust against branches that are only poorly supported by the biological data.
Indeed, to be compatible, all branches must be represented as-is, even when they are weakly supported (see \cref{fig:exampleABCD} for an example).

\begin{figure}[t]
  \centering
  \scalebox{.89}{\begin{tikzpicture}[xscale=.6, yscale=.5]
    \node at (0,-6) {(A)};
    \node[smallvertex] (00) at (0,0) {};
    \nextnode{000}{00}{-135:2}{revarc}
    \nextnode[vertex, xshift=6pt]{0000}{000}{-135:2}{revarc}
    \nextnode[leaf, xshift=12pt, label=below:$a$]{l0}{0000}{-135:2}{revarc}
    \nextnode[leaf, xshift=-12pt, label=below:$b$]{l1}{0000}{-45:2}{revarc}

    \nextnode[reti, label=above:$r$]{0001}{000}{-45:2}{revarc}
    \nextnode[leaf, label=below:$c$]{l2}{0001}{-90:{sqrt(2)}}{revarc}

    \nextnode{001}{00}{-45:2}{revarc}
    \nextnode[leaf, label=below:$d$]{l3}{001}{-90:{2*sqrt(2)}}{revarc}

    \draw[arc] (001) -- (0001);
  \end{tikzpicture}
  \hspace{3ex}
  \begin{tikzpicture}[xscale=.32, yscale=.5]
    \node at (0,-6) {(B)};
    \node[smallvertex] (00) at (0,0) {};
    \nextnode{000}{00}{-135:2}{revarc}
    \nextnode{0000}{000}{-135:2}{revarc}
    \nextnode[leaf, label=below:$a$]{l0}{0000}{-135:2}{revarc}
    \nextnode[leaf, label=below:$b$]{l1}{0000}{-45:2}{revarc}
    \nextnode[leaf, label=below:$c$]{l2}{000}{-45:4}{revarc}
    \nextnode[leaf, label=below:$d$]{l3}{00}{-45:6}{revarc}
  \end{tikzpicture}
  \hspace{3ex}
  \begin{tikzpicture}[xscale=.45, yscale=.5]
    \node at (0,-6) {(C)};
    \node[smallvertex] (00) at (0,0) {};
    \nextnode[vertex, label=left:$x$]{000}{00}{-135:2}{revarc}
    \nextnode[leaf, xshift=20pt, label=below:$a$]{l0}{000}{-135:4}{revarc}

    \nextnode[vertex, label=right:$y$]{0001}{000}{-45:2}{revarc}
    \nextnode[leaf, xshift=-8pt, label=below:$c$]{l2}{0001}{-45:2}{revarc}
    \nextnode[leaf, xshift=8pt, label=below:$b$]{l1}{0001}{-135:2}{revarc}

    \nextnode[leaf, xshift=-20pt, label=below:$d$]{l3}{00}{-45:6}{revarc}
  \end{tikzpicture}
  \hspace{3ex}
  \begin{tikzpicture}[xscale=.45, yscale=.5]
    \node at (0,-6) {(D)};
    \node[smallvertex] (00) at (0,0) {};
    \nextnode[vertex, label=left:$z$]{000}{00}{-135:2}{revarc}
    \nextnode[leaf, xshift=20pt, label=below:$a$]{l0}{000}{-135:4}{revarc}

    \nextnode[leaf, xshift=-8pt, label=below:$c$]{l2}{000}{-45:4}{revarc}
    \nextnode[leaf, xshift=4pt, label=below:$b$]{l1}{000}{-90:{2*sqrt(2)}}{revarc}

    \nextnode[leaf, xshift=-20pt, label=below:$d$]{l3}{00}{-45:6}{revarc}
  \end{tikzpicture}}
  \caption{(A)~Example of a phylogenetic network~$N$ with one reticulation~$r$, and
  (B) a tree~$T_B$ that is (firmly) displayed by $N$ (``compatible'' with $N$).
  The tree~$T_C$ depicted in (C) is \emph{not} firmly displayed by $N$.
  However, if the branch~$xy$ in $T_C$ has low support in the biological data used to construct~$T_C$,
  it is possible that~$T_C$ does not reflect the true evolutionary history between taxa~$a$, $b$, and $c$.
  To avoid such artifacts, branches with low support are usually contracted, resulting in the tree~$T_D$, depicted in~(D)\@,
  which now represents exactly the information that is well supported in the data.
  The information represented by $T_D$ is now consistent with the information represented by~$N$,
  so we would like to say that $T_D$ is ``compatible'' with $N$, even though~$N$ does not firmly display $T_D$.
  This motivates the formulation of ``soft containment''.
  Indeed, $T_D$ is softly displayed by $N$
  (since $T_B$~is a binary resolution of $T_D$ that is firmly displayed by~$N$).}\label{fig:exampleABCD}
\end{figure}

To address this limitation, a relaxed version of \TC{}, called \STC{},
was introduced by Bentert, Mal\'ik, and Weller \cite{BMW18}.
This relaxed framework allows ``undoing'' contractions of poorly supported branches
(contracting low-support branches is a common preprocessing step in network analysis),
effectively allowing uncertainty in the form of so-called ``soft polytomies''.
The corresponding notion of ``soft display'' (see \cref{def:fdsd})
permits us to resolve each high-degree vertex into a binary tree.
Then, the resulting binary resolution of the network ought to firmly display the resulting binary resolution of the tree.

Although \STC{}, like \TC{}, remains NP-hard~\cite{BW21}, we aim to solve even large instances in practice.
To achieve this, we leverage the fact that phylogenetic networks occurring in empirical studies are expected to exhibit a certain degree of tree-likeness,
which can be effectively quantified using the recently introduced measure of tree-likeness called ``scanwidth''~\cite{BSW20}.
Our algorithm solves \STC{} in $2^{\bigO(\outDelta{T} \cdot k \cdot \log(k))} \cdot n^{\bigO(1)}$ time,
where $k \ceq \sw(\Gamma) + \outDelta{N}$,
with $\outDelta{T}$ and $\outDelta{N}$ denoting the maximum out-degrees in the tree and the network,
respectively, and $\sw(\Gamma)$ denoting the scanwidth of a tree extension that is given as part of the input.
(Tree extensions resemble tree decompositions and will be defined formally later.)
The algorithm first makes the network binary through ``stretching'' and ``in-splitting''.
This special case (where the network is binary but the tree need not be) is then solved via bottom-up dynamic programming
along the given tree extension.

\section{Related Work}

On a binary network $N$ and a binary tree $T$,
the \TC{} problem can be solved in $2^{\bigO(t^2)} \cdot \card{A(N)}$ time,
where $t$ denotes the treewidth of $N$ and $A(N)$ is its arc set \cite{IJW23}.
This special case of \TC{} is also a special case of \STC{},
as the notions of ``firm'' and ``soft'' display coincide when $N$ and $T$ are both binary~\cite{BMW18}.
Additionally, \STC{} is a generalization of the \CC{} problem \cite{BMW18}.
Van Iersel, Semple, and Steel \cite{journal-ISS10} show that the latter problem
can be solved in polynomial time on binary networks of constant ``level''.
Note that level-$c$ networks have scanwidth at most $c + 1$ \cite{Hol23}.
A comparison of the level, scanwidth, and treewidth of real-world phylogenetic networks
was conducted by Holtgrefe \cite{Hol23},
concluding that treewidth is not much smaller than scanwidth in practice.
He further points out that ``edge-treewidth'' by Magne et al.~\cite{MPS+23} is the undirected analogue of scanwidth.
We also remark that \STC{} is conceptually similar to \MC{}:
The question of whether a binary network $N$ softly displays a possibly non-binary tree~$T$
is equivalent to asking if $T$ can be obtained from a subtree of $N$ through arc contraction (while preserving leaf-labels).
In fact, the reduction of general \STC{} to this special case is the subject of \cref{sec:arbnet}.
Janssen and Murakami \cite{JM21} study a similar containment notion based on arc contraction.

\section{Preliminaries}

\paragraph{Directed Graphs.}

A \emph{directed graph} (or \emph{digraph} for short) is an ordered pair $D = (V, A)$
that consists of a non-empty, finite set $V$ of \emph{vertices}
and a set $A \subseteq \{(u, v) \in V^2 \mid u \neq v\}$ of \emph{arcs}.
We also use the notation $V(D) \ceq V$ and $A(D) \ceq A$.
The in-neighbors of a vertex $v \in V$ are referred to as its \emph{parents} (denoted $\inV_D (v)$),
and its out-neighbors are referred to as its \emph{children} (denoted $\outV_D (v)$).
The set of incoming arcs of $v$ is denoted $\inA_D (v)$,
and the set of outgoing arcs of $v$ is denoted $\outA_D (v)$.
The \emph{in-degree} of $v$ is $\indeg_D (v) \ceq \card{\inV_D (v)}$,
the \emph{out-degree} of $v$ is $\outdeg_D (v) \ceq \card{\outV_D (v)}$,
and the \emph{degree} of $v$ is $\deg_D (v) \ceq \indeg_D (v) + \outdeg_D (v)$.
The maximum out-degree of any vertex in $D$ is denoted $\outDelta{D}$.
If the maximum in-degree and the maximum out-degree of $D$ are both at most $2$,
then $D$ is called \emph{binary}.
Note that this definition permits a vertex in a binary digraph to have in-degree $1$ and out-degree $1$ simultaneously,
and also in-degree $2$ and out-degree $2$ simultaneously,
making it slightly more permissive than the definition often used in phylogenetics
that requires every non-root, non-leaf vertex in a binary digraph to have degree exactly $3$.
We define the terms \emph{root} and \emph{leaf} in a digraph as follows:
A \emph{root} of $D$ is a vertex with in-degree~$0$,
and $D$ is called \emph{rooted} if it has exactly one root (denoted $\getroot(D)$).
A \emph{leaf} of $D$ is a vertex with out-degree~$0$,
and the set of all leaves of~$D$ is denoted $L(D)$.
We use ${\geq_D} \subseteq V^2$ to denote the reachability relation of~$D$
(so $u \geq_D v$ means that there is a directed path from $u$ to $v$, possibly having length $0$).
The irreflexive kernel of $\geq_D$ is denoted $>_D$.
We say that $D$ is \emph{weakly connected} if the underlying graph of $D$ is connected.
The result of removing $v$ from $D$ is denoted $D - v$.
We also use the notation $D - U$ when removing a set of vertices $U$ from $D$.
If a digraph $D'$ can be obtained by removing vertices and arcs from $D$,
then $D'$ is a \emph{subdigraph} of $D$, denoted $D' \subseteq D$,
and $D'$ is called \emph{leaf-monotone} if, additionally, $L(D')\subseteq L(D)$.
The subdigraph \emph{induced} by a non-empty set $U \subseteq V$ is $D[U] \ceq (U, \{(u, v) \in A \mid u, v \in U\})$.
When \emph{subdividing} an arc $(u, v) \in A$ with a new vertex $w \notin V$,
the arc $(u, v)$ is removed from $D$, while the vertex $w$ and the arcs $(u, w)$ and $(w, v)$ are added.
The inverse operation, removing~$w$ with its incident arcs and adding $(u, v)$, is called \emph{suppression} of $w$.
When \emph{contracting} $(u, v)$,
the parents of $v$, except for $u$, become parents of $u$,
and the children of $v$, except for possibly $u$, become children of $u$;
the vertex $v$ and its incident arcs are then removed.
If $D$ can be obtained from a digraph $D'$ by subdividing arcs,
then $D$ is a \emph{subdivision} of $D'$.
The \emph{union} of two digraphs $D$ and~$D'$ is $D \cup D' \ceq (V(D) \cup V(D'), A(D) \cup A(D'))$.
A \emph{directed acyclic graph} (DAG) is a digraph without directed cycles.
An \emph{out-tree} is a rooted DAG in which every vertex has in-degree at most $1$.

\paragraph{Phylogenetic Networks.}

A (\emph{phylogenetic}) \emph{network} is a rooted DAG
whose root has out-degree at least~$2$
and in which each non-root vertex either has in-degree~$1$ (and is called a \emph{tree vertex})
or has out-degree~$1$ (and is called a \emph{reticulation}), but not both (see \cref{fig:exampleABCD} (A))\@.
The root of a network is also considered to be a tree vertex.
If a network has no reticulations (that is, it is an out-tree), then it is called a (\emph{phylogenetic})~\emph{tree}.

\paragraph{Scanwidth.}

Undoubtedly, one of the most successful concepts in algorithmic graph theory is the notion of \emph{treewidth}~\cite{Bodlaender93}.
While some problems in phylogenetics yield efficient algorithms on networks of bounded treewidth,
the ``undirected nature'' of treewidth often makes dynamic programming prohibitively complicated or outright unimplementable~\cite{SW22}.
Since phylogenetic networks exhibit a ``natural flow of information'' along the direction of the arcs,
Berry, Scornavacca, and Weller~\cite{BSW20} introduced a width measure similar to treewidth, but respecting the direction of the arcs
(see \cite{Reed99,Safari05} for previous efforts to port the notion of treewidth to directed graphs), called ``scanwidth''.
The idea is to ``inscribe'' the arcs of the network~$N$ into an out-tree~$\Gamma$ (called a ``tree extension'') on $V(N)$,
such that all network arcs are directed away from the root of $\Gamma$
(see \cref{fig:sw}).
This structure allows a bottom-up dynamic-programming approach
whose exponential part depends on the size of the set~$\GW_v (\Gamma)$ of arcs of~$N$ that would be cut when slicing just above any vertex~$v$ of $\Gamma$.

We now give a more comprehensive account of scanwidth and define it formally.
Let $D$ be a DAG\@.
The purpose of scanwidth is to measure how similar $D$ is to an out-tree.
This is motivated by the observation that algorithmic problems can often be solved efficiently on (out-)trees,
with the hope that this generalizes to DAGs that have low scanwidth (and are, thus, ``similar'' to out-trees).
Specifically, low scanwidth is supposed to ensure that bottom-up dynamic programming works like on out-trees:
Solutions for separate subtrees can be computed independently,
and these solutions can be combined efficiently once the parent that connects the subtrees is encountered.
Since $D$ is an arbitrary DAG and not necessarily an out-tree,
we first need an appropriate notion of ``subtree''.
For scanwidth, the subtrees to consider are those of a \emph{tree extension} of $D$,
which is an out-tree $\Gamma$ on $V(\Gamma) = V(D)$ with $A(D) \subseteq {>_\Gamma}$.
That is, the out-tree $\Gamma$ has the same vertex set as $D$
and, for every arc $(u, v)$ of $D$, there is a directed path from $u$ to $v$ in $\Gamma$.
This ensures that reachability in $D$ implies reachability in $\Gamma$
(but there is generally no direct correspondence between arcs of $D$ and arcs of $\Gamma$).
Note that tree extensions generalize topological orderings,
and so every DAG has at least one tree extension
(which may be a directed path corresponding to a topological ordering).

So we want to solve algorithmic problems on $D$ via bottom-up dynamic programming along $\Gamma$,
where separate subtrees of $\Gamma$ represent parts of $D$ that should be solved independently.
Crucially, we need to be able to extend solutions for small subtrees to obtain solutions for larger subtrees.
For this to work efficiently, we want the vertices of each subtree of $\Gamma$ to represent a part of $D$
that is mostly ``self-contained'' in the sense that only few arcs in $D$ cross into the part from the outside.
For the subtree of $\Gamma$ rooted at a vertex $t \in V(\Gamma)$,
these are precisely the arcs of the set
$\GW^D_t (\Gamma) \ceq \{(u, v) \in A(D) \mid u >_\Gamma t \geq_\Gamma v\}$.\footnote{Scanwidth can alternatively
be defined in terms of topological orderings instead of tree extensions \cite{BSW20}.
The analogous sets are denoted $\SW$.
The notation $\GW$ reflects that tree extensions are often denoted using a Greek Gamma.
We later use the notation $\HW$ for related sets because H is the successor of G in the Latin alphabet.}
We want even the largest of these sets to be small,
and so we define $\sw_D (\Gamma) \ceq \max_{t \in V(\Gamma)} \card{\GW^D_t (\Gamma)}$.
Then, the \emph{scanwidth} of $D$, denoted $\sw(D)$,
is the minimum $\sw_D (\Gamma)$ over all tree extensions $\Gamma$ of $D$.

\begin{figure}[t]
  \centering
  \scalebox{.89}{\begin{tikzpicture}[xscale=.6, yscale=.5]
    \node at (0,{-7.5+sqrt(2)}) {(A)};
    \node[smallvertex] (00) at (0,0) {};
    \nextnode[vertex, label=above:$u$]{000}{00}{-135:2}{revarc}
    \nextnode[vertex, xshift=6pt]{0000}{000}{-135:2}{revarc}
    \nextnode[leaf, xshift=12pt, label=below:$a$]{l0}{0000}{-135:2}{revarc}
    \nextnode[leaf, xshift=-12pt, label=below:$b$]{l1}{0000}{-45:2}{revarc}

    \nextnode[reti, label=above:$r$]{0001}{000}{-45:2}{revarc}
    \nextnode[leaf, label=below:$c$]{l2}{0001}{-90:{sqrt(2)}}{revarc}

    \nextnode[vertex, label=above:$v$]{001}{00}{-45:2}{revarc}
    \nextnode[leaf, label=below:$d$]{l3}{001}{-90:{2*sqrt(2)}}{revarc}

    \draw[arc] (001) -- (0001);
  \end{tikzpicture}
  \hspace{3ex}
  \begin{tikzpicture}[xscale=.6, yscale=.5]
    \node at (0,-7.5) {(B)};
    \node[smallvertex] (Grt) at (0,0) {};
    \foreach [count=\i from 0] \s/\l in {/v, /u, reti/r}
      \node[vertex, small\s, label=left:$\l$] (v\i) at ($(Grt)+(-90:{sqrt(2)*(\i+1)})$) {};
    \node[smallvertex] (v3) at ($(v1)+(-135:2)$) {};
    \node[smallleaf, label=below:$a$] (Gla) at ($(v3)+(-135:2)$) {};
    \node[smallleaf, label=below:$b$] (Glb) at ($(v3)+(-90:{sqrt(2)})$) {};
    \node[smallleaf, label=below:$c$] (Glc) at ($(v2)+(-90:{sqrt(2)})$) {};
    \node[smallleaf, label=below:$d$] (Gld) at ($(v0)+(-67.5:{3*sqrt(2)/sin(67.5)})$) {};
    \begin{pgfonlayer}{background}
      \draw[lightgray, line width=6pt] (Glc.center) -- (Grt.center);
      \foreach \u/\l in {v0/d, v3/b, v1/a} \draw[lightgray, line width=6pt] (\u.center) -- (Gl\l.center);
    \end{pgfonlayer}
  \end{tikzpicture}
  \hspace{3ex}
  \begin{tikzpicture}[xscale=.6, yscale=.5]
    \node at (0,-7.5) {(C)};
    \node[smallvertex] (Grt) at (0,0) {};
    \foreach [count=\i from 0] \s/\l in {/, /, reti/}
      \node[vertex, small\s, label=left:$\l$] (v\i) at ($(Grt)+(-90:{sqrt(2)*(\i+1)})$) {};
    \nextnode[vertex]{v3}{v1}{-135:2}{revarc}
    \nextnode[leaf, label=below:$a$]{Gla}{v3}{-135:2}{revarc}
    \nextnode[leaf, label=below:$b$]{Glb}{v3}{-90:{sqrt(2)}}{revarc}
    \nextnode[leaf, label=below:$c$]{Glc}{v2}{-90:{sqrt(2)}}{revarc}
    \nextnode[leaf, label=below:$d$]{Gld}{v0}{-67.5:{3*sqrt(2)/sin(67.5)}}{revarc}
    \foreach \u/\v/\b/\s in {Grt/v0/0/, Grt/v1/18/, v0/v2/-18/, v1/v2/0/}
      \draw[arc, \s] (\u) to[bend right=\b] (\v);
    \begin{pgfonlayer}{background}
      \draw[lightgray, line width=6pt] (Glc.center) -- (Grt.center);
      \foreach \u/\l in {v0/d, v3/b, v1/a} \draw[lightgray, line width=6pt] (\u.center) -- (Gl\l.center);
    \end{pgfonlayer}
  \end{tikzpicture}
  \hspace{3ex}
  \begin{tikzpicture}[xscale=.6, yscale=.5]
    \node at (0,-7.5) {(D)};
    \node[smallvertex] (Grt) at (0,0) {};
    \foreach [count=\i from 0] \s/\l in {/, /, reti/}
      \node[vertex, small\s, label=left:$\l$] (v\i) at ($(Grt)+(-90:{sqrt(2)*(\i+1)})$) {};
    \nextnode[vertex]{v3}{v1}{-135:2}{revarc}
    \nextnode[leaf, label=below:$a$]{Gla}{v3}{-135:2}{revarc}
    \nextnode[leaf, label=below:$b$]{Glb}{v3}{-90:{sqrt(2)}}{revarc}
    \nextnode[leaf, label=below:$c$]{Glc}{v2}{-90:{sqrt(2)}}{revarc}
    \nextnode[leaf, label=below:$d$]{Gld}{v0}{-67.5:{3*sqrt(2)/sin(67.5)}}{revarc}
    \foreach \u/\v/\b/\s in {Grt/v0/0/, Grt/v1/18/, v0/v2/-18/, v1/v2/0/}
      \draw[arc, \s] (\u) to[bend right=\b] (\v);
    \begin{pgfonlayer}{background}
      \draw[lightgray, line width=6pt] (Glc.center) -- (Grt.center);
      \foreach \u/\l in {v0/d, v3/b, v1/a} \draw[lightgray, line width=6pt] (\u.center) -- (Gl\l.center);
    \end{pgfonlayer}
    \draw[ultra thick, red!80!black] (-1.6,-3.8) -- (-0.6,-3.8);
    \draw[ultra thick, red!80!black] (-0.5,-3.8) -- (0.5,-3.8);
    \draw[ultra thick, red!80!black] (0.6,-3.8) -- (1.6,-3.8);
  \end{tikzpicture}}
  \caption{(A) Example network $N$ that has scanwidth $2$.
  (B) Possible tree extension $\Gamma$ of $N$.
  The arcs of $\Gamma$ are thick and gray.
  (C) Tree extension $\Gamma$ with the arcs of $N$ drawn in black.
  (D) Snapshot of scanner lines moving along $\Gamma$.
  To give an example,
  the two arcs of $N$ cut by the central scanner line just above the reticulation $r$ constitute the set $\GW^N_r (\Gamma)$.
  The set $\HW^N_u (\Gamma)$ consists of the three arcs of $N$ cut by the left and the central scanner line.}\label{fig:sw}
\end{figure}

Scanwidth can be visualized via ``scanner lines'',
explaining its name:
First, lay out the vertices of $D$ according to $\Gamma$,
where the root of $\Gamma$ is at the top and the leaves of $\Gamma$ are at the bottom.
Second, draw the arcs of $D$.
The definition of a tree extension ensures that all arcs of $D$ now point downward.
Third, imagine horizontal scanner lines,
one above each leaf of $\Gamma$.
Every scanner line cuts through some set of arcs of $D$.
We observe which arcs are cut as a scanner line moves through the drawing from bottom to top,
where distinct scanner lines merge into one scanner line once they meet at a vertex.
Indeed, the arcs cut by a scanner line that is directly above a vertex $t \in V(\Gamma)$
are precisely those of the set $\GW^D_t (\Gamma)$.
We will also be interested in the arcs that are cut by scanner lines directly \emph{below} the vertex $t$.
These are given by the set $\HW^D_t (\Gamma) \ceq \{(u, v) \in A(D) \mid u \geq_\Gamma t >_\Gamma v\}$.
The scanner-line visualization is depicted in \cref{fig:sw}.

For the purpose of dynamic programming,
it will be helpful to assume that we are given a tree extension with special structure:
If $D$ is rooted, then a \emph{canonical} tree extension $\Gamma$ of $D$
is such that $D[\{v \in V(D) \mid t \geq_\Gamma v\}]$ is weakly connected for all $t \in V(\Gamma)$~\cite{Hol23};
that is, the part of $D$ represented by the vertices of the subtree of $\Gamma$ that is rooted at $t$
is always weakly connected.
This is desirable because it guarantees that
$L(\Gamma) = L(D)$ and $\outdeg_\Gamma (v) \leq \outdeg_D (v)$ for each $v \in V(D)$~\cite{BSW20}.
In particular, if $D$ is binary, then so is $\Gamma$.
These properties hold in a canonical tree extension
because a vertex $v$ with greater out-degree in $\Gamma$ than in $D$
would be the root of a subtree of $\Gamma$
corresponding to a part of $D$ that is not weakly connected.
The assumption that a given tree extension is canonical comes at no cost with respect to scanwidth,
and only minor cost in running time:
Given a non-canonical tree extension $\Gamma$ of a rooted DAG $D$,
a canonical tree extension $\Gamma'$ of $D$ with $\sw_D (\Gamma') \leq \sw_D (\Gamma)$
can be constructed in polynomial time~\cite{Hol23}.

Assuming that no ambiguity arises, we may omit subscripts or superscripts.

\section{Deciding Soft Tree Containment}

In this paper, we develop a parameterized algorithm for \STC{}.
Informally speaking, this problem, first considered by Bentert, Mal\'ik, and Weller~\cite{BMW18}, asks us to decide
if a given phylogenetic tree can be embedded in a given phylogenetic network
while taking into account uncertainty in the form of soft polytomies (high-degree vertices to be resolved into a binary tree).
We build up to its formal definition (\cref{prob:stc}) in a way that differs slightly from the original.
First, we adapt the notion of \emph{splitting} \cite{BW21,JM21}.

\begin{definition}
  Let $N$ be any network,
  let $v \in V(N)$ have out-degree at least $3$,
  and let $u, w \in \outV_N (v)$ be two distinct children of $v$.
  Then, \emph{out-splitting} $u$ and $w$ from $v$
  produces the network $N' \ceq (V(N) \uplus \{x\}, A')$,
  where $x \notin V(N)$ is a new vertex and
  $A' \ceq (A(N) \setminus \{(v, u), (v, w)\}) \cup \{(v, x), (x, u), (x, w)\}$.
\end{definition}

\noindent
Exhaustive out-splitting enables us to resolve a high-out-degree tree vertex into a binary out-tree.
Its dual for high-in-degree reticulations is defined below.
The definitions are illustrated in \cref{fig:binres}.

\begin{definition}
  Let $N$ be any network,
  let $v \in V(N)$ have in-degree at least $3$,
  and let $u, w \in \inV_N (v)$ be two distinct parents of $v$.
  Then, \emph{in-splitting} $u$ and $w$ from $v$
  produces the network $N' \ceq (V(N) \uplus \{x\}, A')$,
  where $x \notin V(N)$ is a new vertex and
  $A' \ceq (A(N) \setminus \{(u, v), (w, v)\}) \cup \{(u, x), (w, x), (x, v)\}$.
\end{definition}

\begin{figure}[t]
  \centering
  \scalebox{.89}{\begin{tikzpicture}[xscale=.8, yscale=.6]
    \node at (0,-6) {(A)};
    \node[smallvertex] (00) at (0,0) {};
    \nextnode[vertex]{000}{00}{-135:2}{revarc}
    \nextnode[leaf, label=below:$a$]{l0}{000}{-90:{2*sqrt(2)}}{revarc}

    \nextnode[reti]{0001}{000}{-45:2}{revarc}
    \nextnode[leaf, label=below:$b$]{l1}{0001}{-90:{sqrt(2)}}{revarc}

    \nextnode[vertex]{001}{00}{-45:2}{revarc}
    \nextnode[leaf, label=below:$c$]{l2}{001}{-90:{2*sqrt(2)}}{revarc}

    \draw[arc] (001) -- (0001);
    \draw[arc] (00) -- (0001);
  \end{tikzpicture}
  \hspace{12ex}
  \begin{tikzpicture}[xscale=.8, yscale=.6]
    \node at (0,-6) {(B)};
    \node[smallvertex] (00) at (0,0) {};
    \nextnode[vertex]{00x}{00}{-45:1}{revarc}
    \nextnode[vertex]{000}{00}{-135:2}{revarc}
    \nextnode[leaf, label=below:$a$]{l0}{000}{-90:{2*sqrt(2)}}{revarc}

    \nextnode[reti]{0001}{000}{-45:2}{revarc}
    \nextnode[leaf, label=below:$b$]{l1}{0001}{-90:{sqrt(2)}}{revarc}

    \nextnode[vertex]{001}{00x}{-45:1}{revarc}
    \nextnode[leaf, label=below:$c$]{l2}{001}{-90:{2*sqrt(2)}}{revarc}

    \draw[arc] (001) -- (0001);
    \draw[arc] (00x) -- (0001);
  \end{tikzpicture}
  \hspace{12ex}
  \begin{tikzpicture}[xscale=.8, yscale=.6]
    \node at (0,-6) {(C)};
    \node[smallvertex] (00) at (0,0) {};
    \nextnode[vertex]{00x}{00}{-45:1}{revarc}
    \nextnode[vertex]{000}{00}{-135:2}{revarc}
    \nextnode[leaf, label=below:$a$]{l0}{000}{-90:{2*sqrt(2)}}{revarc}

    \nextnode[reti]{0001x}{000}{-45:1}{revarc}
    \nextnode[reti]{0001}{0001x}{-45:1}{revarc}
    \nextnode[leaf, label=below:$b$]{l1}{0001}{-90:{sqrt(2)}}{revarc}

    \nextnode[vertex]{001}{00x}{-45:1}{revarc}
    \nextnode[leaf, label=below:$c$]{l2}{001}{-90:{2*sqrt(2)}}{revarc}

    \draw[arc] (001) -- (0001);
    \draw[arc] (00x) -- (0001x);
  \end{tikzpicture}}
  \caption{(A) Non-binary network $N$.
  (B) Possible binary out-resolution of $N$, obtained through out-splitting below the root.
  (C) Possible binary resolution of $N$, obtained through in-splitting and out-splitting.}\label{fig:binres}
\end{figure}

We now introduce terminology for partially and fully resolved networks.
Note that we use a different definition of \emph{binary resolution} than Bentert, Mal\'ik, and Weller \cite{BMW18}.
More specifically, we use in-splitting and out-splitting instead of arc contraction.
This avoids corner cases violating the intuition behind soft polytomies.

\begin{definition}
  Let $N$ be a network.
  Any network that is obtained from $N$ through exhaustive out-splitting
  is called a \emph{binary out-resolution} of $N$,
  and any network that is obtained from $N$ through exhaustive in-splitting
  is called a \emph{binary in-resolution} of $N$.
  A \emph{binary resolution} of $N$ is a binary out-resolution of a binary in-resolution of $N$.
\end{definition}

\noindent
Note that, to obtain a binary resolution,
it does not matter whether in-splitting is performed \emph{before} out-splitting or \emph{after};
the order is arbitrary.

When discussing firm and soft display,
we want to treat leaves as labeled and all other vertices as unlabeled.
This is achieved via the following notion of an isomorphism
that agrees with the identity function on all leaves.

\begin{definition}
  Let $D$ and $D'$ be DAGs with $L(D) = L(D')$.
  Let $\iota \colon V(D) \to V(D')$ be an isomorphism between $D$ and $D'$
  with $\iota |_{L(D)} = \id_{L(D)}$,
  where $\iota |_{L(D)} \colon L(D) \to V(D')$ denotes the restriction of $\iota$ to $L(D)$.
  Then, $\iota$ is called \emph{leaf-respecting}.
\end{definition}

At this point, we have everything needed to define the central notions of display for this work.
Recall that, when taking a leaf-monotone subdigraph, no new leaves can be created.

\begin{definition}[\cite{BMW18}]\label{def:fdsd}
  Let $N$ be any network and let $T$ be any tree.
  We say that $N$ \emph{firmly displays} $T$
  if there is a leaf-respecting isomorphism
  between a subdivision of $T$ and a leaf-monotone subdigraph of $N$.
  We say that $N$ \emph{softly displays} $T$
  if there is a binary resolution $N'$ of $N$
  as well as a binary resolution $T'$ of $T$
  such that $N'$ firmly displays $T'$.
\end{definition}

We are now ready to define the problem under consideration.
Note that, for technical reasons, we restrict the input slightly more than Bentert, Mal\'ik, and Weller \cite{BMW18};
in particular, networks in our sense cannot contain
any vertex that has both in-degree and out-degree $1$.
We believe the problem to become computationally harder with such vertices permitted,
despite their insignificance in practice.

\problemdefref{Soft Tree Containment}{stc}
{A phylogenetic network $N$ and a phylogenetic tree $T$ with $L(T) \subseteq L(N)$.}
{Does $N$ softly display $T$?}
{\cite{BMW18}}

\noindent
We solve this problem in two steps, corresponding to the next two sections.
In \cref{sec:binnet}, we consider the special case where the given network~$N$ is binary.
In \cref{sec:arbnet}, we describe a reduction mapping an arbitrary network to a binary one.

\begin{figure}[t]
  \centering
  \scalebox{.89}{\begin{tikzpicture}[xscale=.6, yscale=.5]
    \node at (0,-6) {(A)};
    \node[smallvertex] (00) at (0,0) {};
    \nextnode{000}{00}{-135:2}{revarc, bold, orange}
    \nextnode[vertex, xshift=6pt]{0000}{000}{-135:2}{revarc, bold, magenta}
    \nextnode[leaf, xshift=12pt, label=below:$a$]{l0}{0000}{-135:2}{revarc, bold, magenta}
    \nextnode[leaf, xshift=-12pt, label=below:$b$]{l1}{0000}{-45:2}{revarc, bold, darkgreen}

    \nextnode[reti]{0001}{000}{-45:2}{revarc, bold, blue}
    \nextnode[leaf, label=below:$c$]{l2}{0001}{-90:{sqrt(2)}}{revarc, bold, blue}

    \nextnode{001}{00}{-45:2}{revarc}
    \nextnode[leaf, label=below:$d$]{l3}{001}{-90:{2*sqrt(2)}}{revarc}

    \draw[arc] (001) -- (0001);

    \pgfmathanglebetweenpoints{\pgfpointanchor{000}{center}}{\pgfpointanchor{0000}{center}}
    \let\angleofarc\pgfmathresult
    \begin{scope}
      \clip ($(000)+(90+\angleofarc:0.001)$) -- ($(0000)+(90+\angleofarc:0.001)$) --
        ($(0000)+(90+\angleofarc:1)$) -- ($(000)+(90+\angleofarc:1)$) -- cycle;
      \draw[arc, bold, darkgreen] (000) -- (0000);
    \end{scope}
  \end{tikzpicture}
  \hspace{12ex}
  \begin{tikzpicture}[xscale=.32, yscale=.5]
    \node at (0,-6) {(B)};
    \node[smallvertex] (00) at (0,0) {};
    \nextnode{000}{00}{-135:2}{revarc, bold, orange}
    \nextnode[leaf, label=below:$a$]{l0}{000}{-135:4}{revarc, bold, magenta}
    \nextnode[leaf, label=below:$b$]{l1}{000}{-90:{2*sqrt(2)}}{revarc, bold, darkgreen}
    \nextnode[leaf, label=below:$c$]{l2}{000}{-45:4}{revarc, bold, blue}
    \nextnode[leaf, label=below:$d$]{l3}{00}{-45:6}{revarc}
  \end{tikzpicture}
  \hspace{12ex}
  \begin{tikzpicture}[xscale=.32, yscale=.5]
    \node at (0,-6) {(C)};
    \node[smallvertex] (00) at (0,0) {};
    \nextnode{001}{00}{-45:2}{revarc}
    \nextnode[leaf, label=below:$a$]{l0}{00}{-135:6}{revarc}
    \nextnode[leaf, label=below:$b$]{l1}{001}{-135:4}{revarc}
    \nextnode[leaf, label=below:$c$]{l2}{001}{-90:{2*sqrt(2)}}{revarc}
    \nextnode[leaf, label=below:$d$]{l3}{001}{-45:4}{revarc}
  \end{tikzpicture}}
  \caption{The network $N$ depicted in (A) softly displays the tree $T$ depicted in (B)\@;
  they are a $\YES$-instance of \STC{}.
  However, $N$ does not softly display the tree depicted in (C)\@.
  In (A) and (B)\@, the colors of arcs indicate how to embed $T$ into $N$.
  This intuition will be formalized as follows:
  The colored arcs of $T$ form a \emph{downward-closed subforest}
  whose \emph{top-arc set} consists of the orange arc.
  Each arc of this downward-closed subforest is mapped to the path of the corresponding color in $N$
  by a \emph{soft pseudo-embedding}.
  The magenta path and the green path in $N$ are not arc-disjoint;
  however, they are \emph{eventually arc-disjoint}.
  (This is enabled by (SPE\ref{spe:fork-eventually-disjoint}) in \cref{def:spe}.)}\label{fig:stc}
\end{figure}

A $\YES$-instance and a $\NO$-instance of \STC{} are depicted in \cref{fig:stc}.
The figure also illustrates some of the notions defined and used in \cref{sec:binnet}.

\subsection{Soft Display by Binary Networks}\label{sec:binnet}

Van Iersel, Jones, and Weller~\cite{IJW26} present a dynamic-programming algorithm
that uses tree extensions, thereby leveraging low scanwidth, to solve the \TC{} problem,
which asks for firm rather than soft containment.
In the following, we adapt their algorithm for our purposes.
Although our approach closely follows theirs,
the presentation we give here is self-contained.
We will highlight both similarities and differences between the two approaches.
To start with, we use their definition of a \emph{downward-closed subforest} and of \emph{top arcs} in $T$,
which are mapped into the ``bags'' of a given tree extension of $N$.

In contrast to the definition of a phylogenetic network/tree,
it will be more convenient to work with $N$ and $T$ whose roots have out-degree~$1$.
This is for technical reasons, a standard practice, and does not limit practical applicability.
We also assume to be given a tree extension, ideally of low scanwidth, which might be more problematic in practice,
but it allows us to separate the problem of computing/approximating low-scanwidth tree extensions from the
problem of deciding soft containment.
The following assumption formally introduces these objects,
but its scope is restricted to this section and we drop it for \cref{sec:arbnet}.

\begin{assumption}\label{as:NTGamma}
  Let $N^*$ be any binary phylogenetic network,
  and let $T^*$ be an arbitrary phylogenetic tree with $L(T^*) = L(N^*)$.
  We construct the rooted DAG
  $N \ceq (V(N^*) \uplus \{\rho_N\},\allowbreak A(N^*) \cup \{(\rho_N, \getroot(N^*))\})$
  by attaching a new root $\rho_N \notin V(N^*)$, and we construct the out-tree
  $T \ceq (V(T^*) \uplus \{\rho_T\},\allowbreak A(T^*) \cup \{(\rho_T, \getroot(T^*))\})$
  by attaching a new root $\rho_T \notin V(T^*)$.
  Also, let $\Gamma$ be a canonical tree extension of $N$.
\end{assumption}

We generalize the strict reachability relation defined earlier
in order to define reachability between arcs and vertices/arcs.
The new reachability relation agrees with the old relation when comparing two vertices.

\begin{definition}[\cite{IJW26}]
  Let $D$ be a DAG\@,
  $u, v \in V(D)$ vertices,
  and $(w, x), (y, z) \in A(D)$ arcs.
  Then, the relation ${>_D} \subseteq (V(D) \uplus A(D))^2$ is defined as follows:
  \begin{enumerate}
    \item We have $u >_D v$ if and only if $u \geq_D v$ and $u \neq v$,
    \item we have $(w, x) >_D (y, z)$ if and only if $x \geq_D y$,
    \item we have $u >_D (y, z)$ if and only if $u \geq_D y$, and
    \item we have $(w, x) >_D v$ if and only if $x \geq_D v$ (read: $v$ is \emph{reachable} from $(w, x)$).
  \end{enumerate}
\end{definition}

\noindent
This relation lets us describe parts of $N$ or $T$ that are ``below'' a specified set of arcs.
In particular, our dynamic-programming table will indicate whether the parts of $T$ that are
below an anti-chain of arcs (called ``top arcs'') can be embedded below a set of arcs in $N$
(see \cref{fig:stc}).

\begin{definition}[\cite{IJW26}]
A \emph{downward-closed subforest} of $T$ is a subdigraph $F$ of~$T$
that does not contain isolated vertices
and whose arc set is such that,
for all arcs $a \in A(F)$ and $b \in A(T)$, if $a >_T b$, then $b \in A(F)$.
An arc $(x, y) \in A(F)$ is called a \emph{top arc} of $F$ if $\indeg_F (x) = 0$.
The set $S \subseteq A(F)$ that contains exactly the top arcs of $F$
is called a \emph{top-arc set} for $L(F)$,
and $F$ is called the \emph{forest below} $S$.
\end{definition}

In contrast to firm containment,
soft containment allows embedding distinct arcs~$(x, y)$ and~$(x, z)$ of~$T$ as
paths in $N$ that may start with the same arcs but separate further down.
We call such paths \emph{eventually arc-disjoint}.

\begin{definition}
Two distinct paths $P, Q \subseteq N$ are called \emph{eventually arc-disjoint} if
they have length at least~$1$ and
(a) $P$ and $Q$ are arc-disjoint or
(b) $P$ and $Q$ start at the same vertex $v \in V(N)$ and $P - v$ and $Q - v$ are eventually arc-disjoint.
\end{definition}

\noindent
We remark that the union of several pairwise eventually arc-disjoint paths that all start with the same vertex
can be used to form an out-tree.
To witness soft containment, we associate each vertex of~$T$ with such a collection of paths in~$N$.
Formalizing this intuition,
we adapt the notion of a \emph{pseudo-embedding}~\cite{IJW26} to work with eventually arc-disjoint paths.
Then, we show that such an embedding of $T$ into $N$ witnesses $N^*$ softly displaying~$T^*$.

\begin{definition}\label{def:spe}
  A \emph{soft pseudo-embedding} of a downward-closed subforest~$F$ of $T$
  into the rooted DAG~$N$ is a function~$\phi$ that maps every arc~$(x, y) \in A(F)$
  to some directed path $\phi(x, y) \subseteq N$ of length at least $1$ in $N$
  and satisfies the following:
  \begin{enumerate}[(SPE1)]
    \item\label{spe:last-to-first} For any two arcs of the form $(x, y), (y, z) \in A(F)$,
      the last vertex of $\phi(x, y)$ is the first vertex of $\phi(y, z)$.
    \item\label{spe:paths-disjoint} For all arcs $(x, y), (x', y') \in A(F)$ with $x \neq x'$,
      the two directed paths $\phi(x, y)$ and $\phi(x', y')$ are arc-disjoint.
    \item\label{spe:fork-eventually-disjoint} For all arcs $(x, y), (x, y') \in A(F)$ with $y \neq y'$,
      the two directed paths $\phi(x, y)$ and $\phi(x, y')$ are eventually arc-disjoint.
    \item\label{spe:leaf} For every arc $(x, \ell) \in A(F)$ with $\ell \in L(T)$,
      the last vertex of $\phi(x, \ell)$ is $\ell$.
  \end{enumerate}
  Also, for any arc $(x, y) \in A(F)$, we use $\phi(y)$ to denote the last vertex of $\phi(x, y)$.
\end{definition}

\noindent
If there is a pseudo-embedding (in the sense of Van Iersel, Jones, and Weller~\cite{IJW26})
of a downward-closed subforest of $T$ into $N$,
then it is a \emph{soft} pseudo-embedding as well,
but not vice versa.
We have to use this more general notion because \STC{} is intentionally more permissive than \TC{}.
In fact, handling these generalized embeddings is the primary novelty of the dynamic-programming algorithm
that we develop in this section (\cref{alg:stc}):
We inherit the idea of computing embeddings of parts of $T$ that are below a top-arc set into parts of $N$,
where we only store how each top arc is to be embedded in a ``bag'' of $\Gamma$,
but the details differ significantly for \STC{}.
Crucially, a soft pseudo-embedding can map two distinct arcs to paths that share some arcs;
the paths might only be \emph{eventually} arc-disjoint.
This changes how we generate the embeddings and, thus, changes the analysis.
Moreover, our generalized embeddings can be more numerous
(compare \cref{lem:vsigbound} with \cite[Lemma 9]{IJW26}).
We prove now that these generalized embeddings indeed capture \STC{}.

\begin{restatable}{lemma}{lemsdiffspe}\label{lem:sdiffspe}
The binary network $N^*$ softly displays the tree $T^*$
if and only if there is a soft pseudo-embedding of $T$ into $N$.
\end{restatable}
\begin{proof}\label{pr:lemsdiffspe}
There are two directions to consider.

\proofright{}
Assume that $N^*$ softly displays~$T^*$.
There exists a leaf-respecting isomorphism~$\iota$
between a subdivision~$T^\dagger$ of a binary resolution of $T^*$
and a leaf-monotone subdigraph $N^\dagger$ of~$N^*$.
Note that, since $N^*$ is binary, it is its own unique binary resolution.
We construct a soft pseudo-embedding $\phi$ of $T$ into $N$ as follows:
\begin{enumerate}
\item For every $(x, y) \in A(T^*)$,
we define $\phi(x, y)$ as the path from $\iota(x)$ to $\iota(y)$ in $N^\dagger$.
\item We define $\phi(\rho_T, \getroot(T^*))$ as a path from $\rho_N$ to $\iota(\getroot(T^*))$ in $N$.
\end{enumerate}
This defines $\phi$ for all arcs of $T$ because $A(T) = A(T^*) \cup \{(\rho_T, \getroot(T^*))\}$.
Moreover, $\phi$ maps each arc of $T$ to a directed path of length at least $1$ in $N$.

We show that (SPE\ref{spe:last-to-first}) is satisfied:
Let $(x, y), (y, z) \in A(T)$.
By construction of $\phi$, the last vertex of $\phi(x, y)$ is $\iota(y)$.
Since $x$ is a parent of $y$ in $T$, we have $y \neq \rho_T$.
Thus, the first vertex of $\phi(y, z)$ is also $\iota(y)$.

We show that (SPE\ref{spe:paths-disjoint}) is satisfied:
Let $(x, y), (x', y') \in A(T)$ with $x \neq x'$.
Because $T$ is an out-tree, we have $y \neq y'$.
The case where $x = \rho_T$ or $x' = \rho_T$ is easy.
Hence, we now assume that $(x, y), (x', y') \in A(T^*)$.
Then, $\phi(x, y), \phi(x', y') \subseteq N^\dagger$.
To infer that $\phi(x, y)$ and $\phi(x', y')$ are arc-disjoint,
it is sufficient to argue that,
in $T^\dagger$, the path from $x$ to $y$ and the path from $x'$ to $y'$ are arc-disjoint,
as $\iota$ is a leaf-respecting isomorphism between $T^\dagger$ and $N^\dagger$.
In $T^*$, the path from $x$ to $y$ and the path from $x'$ to $y'$ are clearly arc-disjoint.
When out-splitting or subdividing, the corresponding paths continue to be arc-disjoint,
as a consequence of $x \neq x'$.

We show that (SPE\ref{spe:fork-eventually-disjoint}) is satisfied:
Let $(x, y), (x, y') \in A(T)$ with $y \neq y'$.
Since $y \neq y'$, we know that $x \neq \rho_T$.
Thus, we have $(x, y), (x, y') \in A(T^*)$, and so $\phi(x, y), \phi(x, y') \subseteq N^\dagger$.
It then suffices to argue that,
in $T^\dagger$, the path from $x$ to $y$ and the path from $x$ to $y'$ are eventually arc-disjoint.
This holds in $T^*$ and continues to hold when out-splitting or subdividing.

We show that (SPE\ref{spe:leaf}) is satisfied:
Let $(x, \ell) \in A(T)$ with $\ell \in L(T)$.
By construction of $\phi$, the last vertex of $\phi(x, \ell)$ is $\iota(\ell) = \ell$,
using the fact that $\iota$ is leaf-respecting.

\proofleft{}
Let $\phi$ be a soft pseudo-embedding of $T$ into $N$.
We perform the following steps:
\begin{enumerate}
\item We construct the out-tree $N' \ceq \bigcup_{a \in A(T)} \phi(a)$.
\item We construct the set $R \ceq \{v \in V(N') \mid \forall u \geq_{N'} v \colon \outdeg_{N'} (u) = 1\}$.
\item We construct the out-tree $N'' \ceq N' - R$.
\end{enumerate}
It is easy to see that $N''$ is a leaf-monotone subdigraph of $N^*$.
We argue that there is a leaf-respecting isomorphism
between a subdivision of a binary resolution of $T^*$ and $N''$.
To this end, we define a function~$\tau$ mapping each non-leaf~$x$ of $T^*$
to the binary out-tree to which the outgoing arcs of $x$ in $T^*$ are mapped by~$\phi$.
Then, each $\tau(x)$ can be obtained by out-splitting and subdividing around~$x$ in $T^*$,
except when $x$ is the root of $T^*$;
we have to handle the root separately using the set $R$.

Formally, for every vertex $x \in V(T^*) \setminus L(T^*)$,
we define $\tau(x) \ceq \bigcup_{a \in \outA_{T^*} (x)} \phi(a)$,
and note that $(\bigcup_{x \in V(T^*) \setminus L(T^*)} \tau(x)) - R = N''$.
Let $x, x' \in V(T^*) \setminus L(T^*)$ be distinct.
The constraint (SPE\ref{spe:paths-disjoint}) implies that $\tau(x)$ and $\tau(x')$ are arc-disjoint.
By (SPE\ref{spe:fork-eventually-disjoint}), we know that $\tau(x)$ is an out-tree.
It is binary, its root is $\phi(x)$, and its leaves are $L(\tau(x)) = \{\phi(y) \mid y \in \outV_{T^*} (x)\}$.
Starting from $T^*$, for every $x \in (V(T^*) \setminus L(T^*)) \setminus \{\getroot(T^*)\}$,
an out-tree corresponding to $\tau(x)$ can be constructed through out-splitting and subdividing;
for $x = \getroot(T^*)$, an out-tree corresponding to $\tau(x) - R$ can be constructed.
\end{proof}

\begin{figure}[t]
  \centering
  \begin{tikzpicture}[xscale=.5, yscale=.6]
    \tikzstyle{B1}=[bold, magenta]
    \tikzstyle{B2}=[bold, darkgreen]
    \tikzstyle{B3}=[bold, blue]
    \tikzstyle{fade}=[opacity=.3]
    \node[smallvertex, fade] (rt) at (0,0) {};
    \nextnode[vertex]{0}{rt}{-90:{sqrt(2)}}{revarc, fade}
    \nextnode[vertex, label=left:$v$]{00}{0}{-135:2}{revarc, B1}
    \nextnode[vertex]{01}{0}{-45:2}{revarc, B3}
    \nextnode[leaf, label=below:$a$]{l0}{00}{-135:2}{revarc}
    \nextnode[leaf, label=below:$c$]{l2}{01}{-45:2}{revarc}
    \nextnode[reti]{001}{00}{-45:2}{revarc}
    \nextnode[leaf, label=below:$b$]{l1}{001}{-90:{sqrt(2)}}{revarc}
    \foreach \u/\v/\s in {01/001/} \draw[arc, \s] (\u) -- (\v);
    \begin{scope}
      \clip ($(0)+(-45:0.001)$) -- ($(00)+(-45:0.001)$) -- ($(00)+(-45:1)$) -- ($(0)+(-45:1)$) -- cycle;
      \draw[arc, B2] (0) -- (00);
    \end{scope}
    \coordinate (cutoff) at ($(00)!.5!(0)$);
    \draw[dashed] ($(cutoff)+(180-10:1)$) --++(0:3.6) node[anchor=west] {$B = \GW_v (\Gamma)$};
    \node at (0,-9) {(A)};

    \node[smallvertex] (Grt) at (10,0) {};
    \foreach [count=\i from 0] \s/\l in {/, /v, /, reti/}
      \node[vertex, small\s, label=left:$\l$] (v\i) at ($(Grt)+(-90:{sqrt(2)*(\i+1)})$) {};
    \foreach \n/\p in {c/v2}{\nextnode[leaf, label=below:$\n$]{Gl\n}{\p}{-45:2}{revarc}}
    \nextnode[leaf, label=below:$a$]{Gla}{v1}{-135:2}{revarc}
    \nextnode[leaf, label=below:$b$]{Glb}{v3}{-90:{sqrt(2)}}{revarc}
    \foreach \u/\v/\b/\s in {Grt/v0/0/, v0/v1/0/B1, v0/v2/-30/B3, v1/v3/30/, v2/v3/0/}
      \draw[arc, \s] (\u) to[bend right=\b] (\v);
    \begin{scope}
      \clip ($(v0)+(180:-0.001)$) -- ($(v1)+(180:-0.001)$) -- ($(v1)+(180:-1)$) -- ($(v0)+(180:-1)$) -- cycle;
      \draw[arc, B2] (v0) -- (v1);
    \end{scope}
    \begin{pgfonlayer}{background}
      \draw[lightgray, line width=6pt] (Glb.center) -- (Grt.center);
      \foreach \u/\l in {v2/c, v1/a} \draw[lightgray, line width=6pt] (\u.center) -- (Gl\l.center);
    \end{pgfonlayer}
    \node at (10,-9) {(B)};

    \node[smallvertex, fade] (Trt) at (16,-2) {};
    \nextnode{w0}{Trt}{-90:{sqrt(2)}}{revarc, fade}
    \foreach \n/\p/\a/\s/\ns in {a/w0/90/B1/, b/w0/45/B2/, c/w0/0/B3/}{
      \nextnode[leaf, \ns, label=below:$\n$]{Tl\n}{\p}{-45-\a:2}{revarc, \s}
    }
    \draw[dashed, rounded corners=6] ($(Tlc)+(45:1.1)$) -- ($(Tla)+(45:1.1)$) --++(-1.6,0) node[anchor=east] {$S$};
    \node at (16,-9) {(C)};
  \end{tikzpicture}
  \caption{Valid signature $[B, S, \psi]$, adapted from \cite{IJW26}.
  (C) Top-arc set $S$ in an out-tree $T$.
  The three arcs of $S$ have distinct colors.
  (A) ``Bag'' $B$ of arcs in $N$.
  The colors indicate which arc of $S$ is mapped to which arc in $B$ by $\psi$.
  Note that $\psi$ maps both the magenta arc and the green arc of $S$ to the same magenta--green arc in $B$.
  (B) Canonical tree extension $\Gamma$ of $N$ (arcs of $\Gamma$ depicted in gray)
  with the arcs of $N$ drawn in.
  The set $\GW_v (\Gamma)$ is the current choice for~$B$.}\label{fig:vsig}
\end{figure}

Recall that we assume to be given a canonical tree extension~$\Gamma$ of $N$ (see \cref{as:NTGamma})
and that we define
$\HW_t (\Gamma) \ceq \{(u, v) \in A(N) \mid u \geq_\Gamma t >_\Gamma v\}$
for every $t \in V(\Gamma)$.
When visualizing scanwidth via scanner lines (see \cref{fig:sw}),
the set $\HW_t (\Gamma)$ may be thought of as containing the arcs of~$N$ cut directly below $t \in V(\Gamma)$,
whereas $\GW_t (\Gamma)$ contains the arcs cut directly above $t$.

In the following, we develop a dynamic-programming algorithm
that computes its table entries (``signatures'') bottom-up along $\Gamma$.
Intuitively, for each vertex~$t \in V(\Gamma)$, each top-arc set~$S$ in~$T$, and each $\psi \colon S \to \GW_t (\Gamma)$,
we set the table entry / signature~$[\GW_t (\Gamma), S, \psi]$ to ``valid''
if the forest below $S$ in $T$ is softly displayed by the network below $\GW_t (\Gamma)$ in such a way
that each arc $a \in S$ is embedded in a path starting with $\psi(a) \in \GW_t (\Gamma)$.
An example of a valid signature is depicted in \cref{fig:vsig}.
Note that $\psi$ is not required to be injective.

\begin{definition}
  A \emph{signature} is an ordered triple $[B, S, \psi]$
  containing a set $B \subseteq A(N)$ of arcs of $N$,
  a top-arc set $S \subseteq A(T)$ for the set of leaves reachable from $B$ in $N$,
  and a function $\psi \colon S \to B$.
  Such a signature $[B, S, \psi]$ is called \emph{valid}
  if there is a soft pseudo-embedding $\phi$ of the forest below $S$ into $N$ such that,
  for every arc $(x, y) \in S$,
  the first arc of the path $\phi(x, y)$ is $\psi(x, y)$;
  then, $\phi$ is a \emph{witness} for the valid signature $[B, S, \psi]$.
\end{definition}

Our dynamic-programming algorithm, \cref{alg:stc}, decides whether $N^*$ softly displays $T^*$ by,
in some sense,
exhaustively generating valid signatures.
Once it has generated them,
it can easily read off the final answer (\cref{lem:vsigtoparc}).
To prove the main result of this section (\cref{prop:stc}),
we establish the correctness of \cref{alg:stc} and bound its running time.
This involves showing that \cref{alg:stc} generates its valid signatures correctly (\cref{lem:vsigleaf,lem:vsigHW,lem:vsigGW})
and that there are not too many valid signatures to generate (\cref{lem:vsigbound}).
We start with the latter.

\begin{restatable}{lemma}{lemvsigbound}\label{lem:vsigbound}
  For each set $B \subseteq A(N)$,
  there are at most $(4 \cdot \card{B})^{\outDelta{T} \cdot \card{B}}$ valid signatures
  of the form $[B, S, \psi]$.
\end{restatable}
\begin{proof}\label{pr:lemvsigbound}
  Let $B \subseteq A(N)$ be a set of arcs.
  For all valid signatures of the form $[B, S, \psi]$,
  it holds that $\card{S} \leq \outDelta{T} \cdot \card{B}$,
  since $\card{\psi^{-1} (u, v)} \leq \outDelta{T}$
  for every arc $(u, v) \in B$.
  This is because, when $\psi$ maps multiple arcs of $S$ to a single arc in $B$,
  these arcs of $S$ must have the same tail in order to permit eventual arc-disjointness.
  Therefore, to construct a valid signature of the form $[B, S, \psi]$,
  the number of choices for $S$
  is at most~$4^{\outDelta{T} \cdot \card{B}}$:
  This bound on the number of top-arc sets of bounded size for a fixed set of leaves
  is derived from a bound on the number of \emph{important separators}
  due to Chitnis, Hajiaghayi, and Marx~\cite[Lemma~4.2]{CHM13}.
  Specifically, from every top-arc set $S$,
  an important $\{\rho_T\}$--$L(F)$ separator in a slight modification of $T$ can be obtained
  using the heads of the arcs in $S$,
  where $F$ is the forest below $S$ (see also \cite[Lemma 4]{IJW26}).
  For each choice of $S$,
  the number of choices for $\psi$
  is at most~$\card{B}^{\card{S}} \leq \card{B}^{\outDelta{T} \cdot \card{B}}$.
  Multiplying $4^{\outDelta{T} \cdot \card{B}}$ with $\card{B}^{\outDelta{T} \cdot \card{B}}$
  gives the promised bound.
\end{proof}

For every non-root vertex~$v \in V(\Gamma)$,
\cref{alg:stc} constructs the set~$\GWS_v$ of pairs $(S, \psi)$ for which $[\GW_v (\Gamma), S, \psi]$ is a valid signature
(and analogously for $\HWS_v$).
The construction has three steps:
First, handling the case where $v$ is a leaf.
Second, constructing $\HWS_v$ when $v$ is not a leaf.
Third, constructing $\GWS_v$ when $v$ is not a leaf.
The next three lemmas imply the steps' correctness.

\begin{restatable}{lemma}{lemvsigleaf}\label{lem:vsigleaf}
  For each leaf $v \in L(\Gamma)$,
  there is exactly one ordered pair $(S, \psi)$
  such that $[\GW_v (\Gamma), S, \psi]$ is a valid signature.
\end{restatable}
\begin{proof}\label{pr:lemvsigleaf}
  Let $v \in L(\Gamma) = L(N) = L(T)$ be a leaf.
  We observe that a signature of the form $[\GW_v (\Gamma), S, \psi]$ is valid
  if and only if
  (a) $S = \inA_T (v)$ and
  (b) $\psi$ maps the unique arc in~$S$ to the unique arc in $\GW_v (\Gamma) = \inA_N (v)$.
\end{proof}

Recall that $N$ is binary.
Since $\Gamma$ is canonical, this implies that it too is binary \cite{BSW20},
and so any vertex has at most two children in $\Gamma$.
For every non-root, non-leaf $v \in V(\Gamma)$,
\cref{alg:stc} constructs $\HWS_v$ by combining pairs from $\GWS_{q_1}$ and $\GWS_{q_2}$,
where $q_1$ and $q_2$ are the children of~$v$ in $\Gamma$.
We define $D_u \ceq D[\{w \in V(D) \mid u \geq_D w\}]$
for each~$u \in V(D)$ in a DAG~$D$,
to refer to subtrees of $T$ and $\Gamma$.
Recall that $\psi |_Z$ is the restriction of $\psi$ to $Z$.

\begin{restatable}{lemma}{lemvsigHW}\label{lem:vsigHW}
  Let $v \in (V(\Gamma) \setminus \{\getroot(\Gamma)\}) \setminus L(\Gamma)$
  have children~$\outV_\Gamma (v) = \{q_1, q_2\}$ in~$\Gamma$,
  where $q_1$ and $q_2$ are not necessarily distinct,
  and let~$[\HW_v (\Gamma), S, \psi]$ be a signature.
  Further, for each~$i \in \{1, 2\}$,
  let~$S_i \ceq \{(x, y) \in S \mid L(T_y) \subseteq L(\Gamma_{q_i})\}$.
  Then, the signature $[\HW_v (\Gamma), S, \psi]$ is valid
  if and only if
  (a)~$S = S_1 \cup S_2$ and
  (b)~$[\GW_{q_i} (\Gamma), S_i, \psi |_{S_i}]$ is a valid signature for each~$i \in \{1, 2\}$.
\end{restatable}
\begin{proof}\label{pr:lemvsigHW}
\proofright{}
Suppose that $[\HW_v (\Gamma), S, \psi]$ is valid,
and let $\phi$ be a witness for it.

First, we show $S \subseteq S_1 \cup S_2$.
Because $S_1 \cup S_2 \subseteq S$ by definition,
this then implies $S = S_1 \cup S_2$.
Let $(x, y) \in S$ be an arc of $T$ and
let $(u, w) \ceq \psi(x, y)$.
Since $(u, w) \in \HW_v (\Gamma)$,
we have $v >_\Gamma w$ and,
hence, $q_i \geq_\Gamma w$ for some $i \in \{1, 2\}$.
In the following, we prove that $L(T_y) \subseteq L(\Gamma_{q_i})$ and, thus, $(x, y) \in S_i$.
To this end,
let $\ell \in L(T_y)$ be a leaf and
let $P \subseteq T_y$ be the $y$--$\ell$ path in $T_y$.
Then, $\bigcup_{a \in A(P)} \phi(a)$ is a path from $\phi(y)$ to $\ell$ in $N$
unless $\phi(y) = \ell$.
Either way, $w \geq_N \phi(y) \geq_N \ell$,
so $q_i \geq_\Gamma w \geq_\Gamma \ell$ and, consequently, $\ell \in L(\Gamma_{q_i})$.

Second, let~$i \in \{1, 2\}$ and
let~$F \subseteq T$ be the forest below $S_i$.
Then, $\phi |_{A(F)}$ is a witness
for the valid signature $[\GW_{q_i} (\Gamma), S_i, \psi |_{S_i}]$.

\proofleft{}
We assume that $S = S_1 \cup S_2$
and that $[\GW_{q_i} (\Gamma), S_i, \psi |_{S_i}]$ is a valid signature for each~$i \in \{1, 2\}$, witnessed by $\phi_i$.
If $q_1 = q_2$,
then $\HW_v (\Gamma) = \GW_{q_1} (\Gamma)$, $S = S_1$, and $\psi = \psi |_{S_1}$,
so $[\HW_v (\Gamma), S, \psi]$ is valid.
If $q_1 \neq q_2$,
then $S_1 \cap S_2 = \nothing$ and $\phi \ceq \phi_1 \uplus \phi_2$ is a witness for $[\HW_v (\Gamma), S, \psi]$.
\end{proof}

\begin{algorithm}[tp]
\caption{Solve \STC{} on $N^*$ and $T^*$.}\label{alg:stc}
\begin{algorithmic}[1]
\Require $N$ (binary), $T$ (not necessarily binary), and $\Gamma$; see \cref{as:NTGamma}.
\Ensure If $N^*$ softly displays $T^*$, then $\YES$, otherwise $\NO$.
\Procedure{SolveSTC}{$N, T, \Gamma$}
\For{$v \in V(\Gamma) \setminus \{\getroot(\Gamma)\}$ in a bottom-up traversal of $\Gamma$}
	\State $\GWS_v \gets \nothing$;
	\State $\HWS_v \gets \nothing$;
	\If{$v \in L(\Gamma)$}
		\State $p \gets \text{parent of $v$ in $T$}$;
		\State $u \gets \text{parent of $v$ in $N$}$;
		\State $\GWS_v \gets \{(\{(p, v)\}, \psi)\}$, where $\psi(p, v) \ceq (u, v)$;
	\Else
		\State $\{q_1, q_2\} \gets \outV_\Gamma (v)$;
			\Comment{Note that $q_1$ and $q_2$ are not necessarily distinct.}
		\If{$q_1 = q_2$}
			\State $\HWS_v \gets \GWS_{q_1}$;
		\Else
			\For{$((S_1, \psi_1), (S_2, \psi_2)) \in \GWS_{q_1} \times \GWS_{q_2}$}\label{line:loop}
				\State Add $(S_1 \cup S_2, \psi_1 \cup \psi_2)$ to $\HWS_v$;
			\EndFor
		\EndIf
		\State \Call{PopulateGWS}{$N, T, \Gamma, v$};
			\Comment{Populate $\GWS_v$ using $\HWS_v$.}
	\EndIf
\EndFor
\State $v \gets \text{child of the root of $N$}$;
\State $a \gets \text{top arc of $T$}$;
\If{$\exists (S, \psi) \in \GWS_v \colon S = \{a\}$}
	\State \Return $\YES$;
\Else
	\State \Return $\NO$;
\EndIf
\EndProcedure
\Procedure{PopulateGWS}{$N, T, \Gamma, v$}
\For{$(S, \psi) \in \HWS_v$}
	\If{$\psi^{-1} (\outA_N (v)) = \nothing$}
		\State Add $(S, \psi)$ to $\GWS_v$;
	\Else
		\If{$\exists x \in V(T) \colon \psi^{-1} (\outA_N (v)) \subseteq \outA_T (x)$}
			\For{$u \in \inV_N (v)$}
				\State Define $\psi' \colon S \to \GW_v (\Gamma)$ such that
				\State $\psi'(x, y) \ceq (u, v)$ for every $(x, y) \in \psi^{-1} (\outA_N (v))$ and
				\State $\psi'(x', y') \ceq \psi(x', y')$ for every $(x', y') \in S \setminus \psi^{-1} (\outA_N (v))$;
				\State Add $(S, \psi')$ to $\GWS_v$;
			\EndFor
		\EndIf
		\If{$\exists y \in V(T) \setminus \{\getroot(T)\} \colon \psi^{-1} (\outA_N (v)) = \outA_T (y)$}
			\State $x \gets \text{parent of $y$ in $T$}$;
			\State $S' \gets (S \setminus \psi^{-1} (\outA_N (v))) \cup \{(x, y)\}$;
			\For{$u \in \inV_N (v)$}
				\State Define $\psi' \colon S' \to \GW_v (\Gamma)$ such that $\psi'(x, y) \ceq (u, v)$ and
				\State $\psi'(x', y') \ceq \psi(x', y')$ for every $(x', y') \in S \setminus \psi^{-1} (\outA_N (v))$;
				\State Add $(S', \psi')$ to $\GWS_v$;
			\EndFor
		\EndIf
	\EndIf
\EndFor
\EndProcedure
\end{algorithmic}
\end{algorithm}

The next lemma concerns the correctness of the procedure \textsc{PopulateGWS}\@.
Given a non-root, non-leaf vertex $v \in V(\Gamma)$,
this procedure constructs the set $\GWS_v$ based on the elements of $\HWS_v$.
Informally speaking, this step from $\HW_v (\Gamma)$ to $\GW_v (\Gamma)$ unlocks the incoming arcs of $v$ in $N$.
There are three ways in which these arcs may be used:
(1) not at all,
(2) to elongate paths of an existing embedding, or
(3) to enlarge the downward-closed subforest to be embedded.

\begin{restatable}{lemma}{lemvsigGW}\label{lem:vsigGW}
Let $v \in (V(\Gamma) \setminus \{\getroot(\Gamma)\}) \setminus L(\Gamma)$ be a vertex and
$[\GW_v (\Gamma), S, \psi]$ a signature.
Then, $[\GW_v (\Gamma), S, \psi]$ is valid
if and only if at least one of the three conditions below is met:
\begin{enumerate}[(1)]
\item $\psi(x, y) \notin \inA_N (v)$ for every $(x, y) \in S$,
and $[\HW_v (\Gamma), S, \psi]$ is a valid signature.
\item $(u, v) \in \psi(S)$ for exactly one vertex $u \in \inV_N (v)$;
further, it holds that
\begin{enumerate}[(2.1)]
\item $X \ceq \psi^{-1} (u, v) \subseteq \outA_T (x)$ for some $x \in V(T)$, and
\item $[\HW_v (\Gamma), S, \psi']$ is a valid signature,
where $\psi' \colon S \to \HW_v (\Gamma)$ is some function with
$\psi'(X) \subseteq \outA_N (v)$ and,
for all $(x', y') \in S \setminus X$,
we have $\psi'(x', y') = \psi(x', y')$.
\end{enumerate}
\item $\psi(x, y) \in \inA_N (v)$ for exactly one $(x, y) \in S$;
further, it holds that
\begin{enumerate}[(3.1)]
\item $y \notin L(T)$, and
\item $[\HW_v (\Gamma), S_{x, y}, \psi']$ is a valid signature,
where $S_{x, y} \ceq (S \setminus \{(x, y)\}) \cup \outA_T (y)$
and $\psi' \colon S_{x, y} \to \HW_v (\Gamma)$ is some function with
$\psi'(\outA_T (y)) \subseteq \outA_N (v)$ and,
for all $(x', y') \in S \setminus \{(x, y)\}$,
we have $\psi'(x', y') = \psi(x', y')$.
\end{enumerate}
\end{enumerate}
\end{restatable}
\begin{proof}\label{pr:lemvsigGW}
\proofright{}
Suppose that $[\GW_v (\Gamma), S, \psi]$ is valid,
and let $\phi$ be a witness for it.
Then, at least one of the following holds, and we consider them individually:
\begin{enumerate}[(a)]
\item $\psi(x, y) \notin \inA_N (v)$ for every arc $(x, y) \in S$.
\item $\psi(x, y) \in \inA_N (v)$ for some arc $(x, y) \in S$ with $\phi(y) \neq v$.
\item $\psi(x, y) \in \inA_N (v)$ for some arc $(x, y) \in S$ with $\phi(y) = v$.
\end{enumerate}

Case (a):
Because $\psi(S) \subseteq \GW_v (\Gamma) \setminus \inA_N (v) \subseteq \HW_v (\Gamma)$,
we know that $[\HW_v (\Gamma), S, \psi]$ is a valid signature
witnessed by $\phi$,
so (1) is met.

Case (b):
Based on $\phi$, we construct a new soft pseudo-embedding $\phi'$ of the forest below $S$ into $N$.
To this end, we distinguish between two kinds of arcs in $S$:
those that $\psi$ maps into $\inA_N (v)$
and those that $\psi$ maps to elsewhere.
Only on the former arcs is $\phi'$ defined differently than $\phi$.
Namely, for every such arc, we remove the first vertex of the path to which the arc is mapped.
We do this because $\inA_N (v)$ and $\HW_v (\Gamma)$ are disjoint.
We now describe the steps formally.
We observe that $\psi(S) \cap \inA_N (v)$ contains a single arc~$(u, v)$.
Consider the set $X \ceq \psi^{-1} (u, v)$, and note that $(x, y) \in X$ and $X \subseteq \outA_T (x)$.
For each arc $(x, z) \in X$,
we define $\phi'(x, z)$ as the path
obtained by removing $u$ from the path $\phi(x, z)$.
For each arc $(x', y') \in A(F) \setminus X$,
where $F$ is the forest below $S$,
we define $\phi'(x', y') \ceq \phi(x', y')$.
Further, we construct $\psi' \colon S \to \HW_v (\Gamma)$ by
setting $\psi'(a)$ to be the first arc of $\phi'(a)$ for all $a \in S$.
Then, $[\HW_v (\Gamma), S, \psi']$ is a valid signature,
and $\phi'$ is a witness for it.
We observe that (2) is met now.

Case (c):
We use that no arc of $S$ other than $(x, y)$ is mapped into $\inA_N (v)$ by $\psi$.
This lets us construct a new soft pseudo-embedding $\phi'$
by simply taking the restriction of $\phi$ to the arcs of the forest below $S_{x, y}$.
Formally, we do the following:
As $\phi(y) = v$ and $v \notin L(\Gamma) = L(N) = L(T)$,
we have $y \notin L(T)$.
The existence of an arc $(x', y') \in S \setminus \{(x, y)\}$ with $\psi(x', y') \in \inA_N (v)$
would lead to violations of (eventual) arc-disjointness,
so no such arc $(x', y')$ exists.
We construct $\phi' \ceq \phi |_{A(F)}$,
where $F$ is the forest below $S_{x, y}$,
and we define $\psi'(a)$ to be the first arc of $\phi'(a)$ for each $a \in S_{x, y}$.
Then, $\phi'$ is a witness for the valid signature $[\HW_v (\Gamma), S_{x, y}, \psi']$.
We can infer that (3) is met.

\medskip
\proofleft{}
We argue the three cases, assuming that (1), (2), or (3) holds, individually.

Case (1):
Let $\phi'$ be a witness for the valid signature $[\HW_v (\Gamma), S, \psi]$.
Then, $\phi'$ is a witness for $[\GW_v (\Gamma), S, \psi]$ as well,
proving its validity.

Case (2):
Let $\phi'$ be a witness for the valid signature $[\HW_v (\Gamma), S, \psi']$.
For each $(x, y) \in X$,
we define $\phi(x, y)$ to be the path
obtained by prepending the vertex $u$ and the arc $(u, v)$ to the path $\phi'(x, y)$.
For each $(x', y') \in A(F) \setminus X$,
where $F$ is the forest below $S$,
we define $\phi(x', y') \ceq \phi'(x', y')$.
Then, $\phi$ is a witness for the valid signature $[\GW_v (\Gamma), S, \psi]$.

Case (3):
Let $\phi'$ be a witness for the valid signature $[\HW_v (\Gamma), S_{x, y}, \psi']$.
We define $\phi(x, y)$ to be the length-$1$ path containing the arc $\psi(x, y)$,
and we define $\phi(x', y') \ceq \phi'(x', y')$ for every $(x', y') \in A(F)$,
where $F$ is the forest below $S_{x, y}$.
Then, $\phi$ witnesses the validity of $[\GW_v (\Gamma), S, \psi]$.
\end{proof}

While the previous lemmas established that \cref{alg:stc} constructs its valid signatures correctly
(that is, the $(S, \psi)$ pairs it generates correspond precisely to the valid signatures of interest),
the next lemma implies that its final answer is correct.

\begin{restatable}{lemma}{lemvsigtoparc}\label{lem:vsigtoparc}
  Let $v \in V(N)$ be the child of the root of $N$,
  and let $a \in A(T)$ be the top arc of $T$.
  Then,
  $N^*$ softly displays $T^*$
  if and only if
  there exists a valid signature of the form $[\GW_v (\Gamma), \{a\}, \psi]$.
\end{restatable}
\begin{proof}\label{pr:lemvsigtoparc}
  \proofright{}
  Suppose that $N^*$ softly displays $T^*$.
  From the proof of \cref{lem:sdiffspe},
  we deduce that there is a soft pseudo-embedding $\phi$ of $T$ into $N$
  such that the first arc of the path $\phi(a)$ is $(\rho_N, v)$.
  We define the function $\psi(a) \ceq (\rho_N, v)$.
  Then, $[\GW_v (\Gamma), \{a\}, \psi]$ is a valid signature,
  and $\phi$ is a witness for it.

  \proofleft{}
  Let $[\GW_v (\Gamma), \{a\}, \psi]$ be a valid signature,
  and let $\phi$ be a witness for it.
  Then, $\phi$ is a soft pseudo-embedding of $T$ into $N$.
  Thus, $N^*$ softly displays $T^*$ by \cref{lem:sdiffspe}.
\end{proof}

We now piece the previous lemmas together to conclude this section.

\begin{proposition}\label{prop:stc}
  Given $N$, $T$, and $\Gamma$ fulfilling \cref{as:NTGamma},
  \cref{alg:stc} decides whether $N^*$ softly displays~$T^*$
  in $\bigO^* (2^{\outDelta{T} \cdot (\sw(\Gamma) + 1) \cdot \log_2 (4 \cdot \sw(\Gamma) + 4)})$~time.
\end{proposition}
\begin{proof}
Let $v \in V(\Gamma) \setminus \{\getroot(\Gamma)\}$ be a vertex.
Whenever \cref{alg:stc} adds a pair $(S, \psi)$ to $\GWS_v$ (or $\HWS_v$),
it is true that $[\GW_v (\Gamma), S, \psi]$ (or $[\HW_v (\Gamma), S, \psi]$) is a signature.
From \cref{lem:vsigleaf,lem:vsigHW,lem:vsigGW},
it follows that \cref{alg:stc} generates precisely every such signature that is valid.
Then, the correctness of \cref{alg:stc} is a consequence of \cref{lem:vsigtoparc}.
The bound on the running time may be deduced as follows:
Because $N$ is binary,
we have $\card{\HW_v (\Gamma)} \leq \card{\GW_v (\Gamma)} + 1 \leq \sw(\Gamma) + 1$.
Hence, by \cref{lem:vsigbound},
the sets $\GWS_v$ and $\HWS_v$ contain at most
$2^{\outDelta{T} \cdot (\sw(\Gamma) + 1) \cdot \log_2 (4 \cdot \sw(\Gamma) + 4)}$
elements each.
This also bounds the number of iterations of the loop in \cref{line:loop},
as each iteration adds a new pair to $\HWS_v$.
\end{proof}

Note that \cref{alg:stc} decides \STC{} without explicitly constructing a full soft pseudo-embedding of $T$ into $N$.
However, it is not difficult to adapt the algorithm so as to output such a soft pseudo-embedding on a $\YES$-instance:
Let $v \in (V(\Gamma) \setminus \{\getroot(\Gamma)\}) \setminus L(\Gamma)$.
Simply have the algorithm store, for each pair $(S, \psi) \in \GWS_v$,
a single pair from $\HWS_v$ from which $(S, \psi)$ was constructed.
Similarly, have the algorithm store, for each pair in $\HWS_v$,
the one or two pairs from which it was constructed.
In case of a $\YES$-instance,
these separately stored pairs would allow the algorithm, at the end,
to construct a soft pseudo-embedding of $T$ into $N$ through backtracing.

\subsection{Soft Display by Arbitrary Networks}\label{sec:arbnet}

In this section,
we devise a reduction that maps any network to a binary one
while preserving the trees that the network softly displays.
The reduction can be performed in polynomial time,
and the resulting increase in the scanwidth of the network can be bounded.
Relying on the considerations of the previous section,
we then obtain our main result (\cref{thm:stc}).
The reduction consists of two steps:
\begin{enumerate}
\item We \emph{stretch} the network.
This means that we modify the network
using gadgets that are conceptually very similar to the \emph{universal networks} by Zhang~\cite{Zha16}
(see \cref{fig:str}, also compare \cite[Figure~4]{Zha16}).
Like these universal networks,
our gadgets consist of a triangular upper part and a lower part that is based on a sorting network.

The triangular upper part is essentially one half of a grid that was slightly modified to obtain a phylogenetic network.
In it, any binary phylogenetic tree (up to some number of leaves) can be embedded.
As each high-out-degree vertex is replaced by a gadget when stretching,
the stretched network represents all possible binary out-resolutions.

However, embedding a binary tree in the triangular upper part
artificially enforces an order on the leaves of the binary tree.
The lower part allows us to undo this.
We choose a different sorting network than Zhang \cite{Zha16}.

\item We compute an arbitrary binary in-resolution of the stretched network through exhaustive in-splitting.
This does not affect which trees are firmly displayed.
\end{enumerate}

The purpose of stretching, defined below, is to get rid of vertices with high out-degree in a network,
so that we can assume the resulting network to have maximum out-degree $2$.

\begin{definition}\label{def:str}
Let $N$ be any arbitrary network,
and let $v \in V(N)$ be a vertex of out-degree $d \ceq \outdeg_N (v)$ at least $3$.
When \emph{stretching} $v$, we modify $N$ as follows:
First, we remove the arcs between $v$ and its children $\outV_N (v) = \{c_1, c_2, \dots, c_d\}$.
Then, we insert a gadget consisting of two parts.

For the first part, we start by inserting a triangular half of a grid rooted at $v$.
The construction proceeds in layers:
The new children of $v$ are two new vertices $u_{2, 1}$ and $u_{2, 2}$ now.
Below this layer, there are $d - 2$ additional layers.
Each layer has an index $i \in \{2, 3, \dots, d\}$ and consists of the vertices $\{u_{i, j} \mid j \in [i]\}$.
For $i < d$, every vertex $u_{i, j}$ has the children $u_{i + 1, j}$ and $u_{i + 1, j + 1}$.
This triangular half-grid can contain some vertices $u_{i, j}$ that simultaneously have in-degree $2$ and out-degree $2$.
We fix this by making the two children of such a vertex $u_{i, j}$ the children of a new vertex $u'_{i, j}$ instead,
and we add an arc from $u_{i, j}$ to $u'_{i, j}$.
Further, we suppress the two vertices $u_{d, 1}$ and $u_{d, d}$ later,
since otherwise they would have in-degree $1$ and out-degree $1$ simultaneously in the final gadget.

For now, we start constructing the second part of the gadget by adding $d$ arcs,
each one from $u_{d, j}$ to $c_j$ for $j \in [d]$.
We subdivide these arcs repeatedly to obtain $d$ paths,
and we add arcs between the paths to obtain a sorting network.
Specifically, we construct $(d - 1)^2$ so-called \emph{crossovers} \cite{Hay16}:
A crossover has indices $i, j \in [d - 1]$ and involves four vertices,
denoted $w_{i, j, k}$ for $k \in [4]$.
These vertices all lie on the $d$ paths that we created through repeated subdivision.
Each crossover has four arcs:
from $w_{i, j, 1}$ to $w_{i, j, 3}$ and to $w_{i, j, 4}$,
as well as from $w_{i, j, 2}$ to $w_{i, j, 3}$ and to $w_{i, j, 4}$.
We place these crossovers so that they resemble a bubble-sort network with $d - 1$ passes (indexed by $i$)
where, per pass, $d - 1$ comparisons (indexed by $j$) are performed.
Each crossover represents one such comparison.

When \emph{stretching} $N$, we stretch every vertex of $N$ that has an out-degree of at least~$3$,
and we denote the resulting network $\str(N)$.
\end{definition}

\noindent
In the appendix,
we provide an equivalent definition where the vertices and the arcs of a stretch gadget are listed explicitly
(\cref{def:strexp}).

\def\rotfig{70}
\def\drawswitch#1#2#3{
  \node (w1#2#3) at (#1) {$w_{#2, #3, 1}$};
  \node (w2#2#3) at ($(w1#2#3)+(0+\rotfig:{sqrt(2)})$) {$w_{#2, #3, 2}$};
  \node (w4#2#3) at ($(w2#2#3)+(-90+\rotfig:2)$) {$w_{#2, #3, 4}$};
  \node (w3#2#3) at ($(w1#2#3)+(-90+\rotfig:2)$) {$w_{#2, #3, 3}$};
  \foreach \u in {1,2}{
    \foreach \v in {3,4}{
      \draw (w\u#2#3) -- (w\v#2#3);
  }}
}
\begin{figure}[t]
\centering
\scalebox{.84}{\begin{tikzpicture}[xscale=.7, yscale=.6, every path/.style={-latex}]
  \node (v) at (0,0) {$v$};
  \node (u21) at ($(v)+(-135+\rotfig:2)$) {$u_{2, 1}$};
  \node (u22) at ($(v)+(-45+\rotfig:2)$) {$u_{2, 2}$};
  \node (u32) at ($(u21)+(-45+\rotfig:2)$) {$u_{3, 2}$};
  \drawswitch{$(u21)+(-90+\rotfig:{2+sqrt(2)})$}{1}{1}
  \drawswitch{$(w411)+(-90+\rotfig:2)$}{1}{2}
  \drawswitch{$(w311)+(-90+\rotfig:6)$}{2}{1}
  \drawswitch{$(w421)+(-90+\rotfig:2)$}{2}{2}
  \node (c1) at ($(w321)+(-90+\rotfig:6)$) {$c_1$};
  \node (c2) at ($(c1)+(0+\rotfig:{sqrt(2)})$) {$c_2$};
  \node (c3) at ($(c2)+(0+\rotfig:{sqrt(2)})$) {$c_3$};
  \foreach \u/\v in {v/u21, v/u22, u21/w111, u21/u32, u22/u32, u22/w212, u32/w211, w321/c1, w322/c2, w422/c3} \draw (\u) -- (\v);
  \foreach \u/\v in {w311/w121, w411/w112, w312/w221, w421/w122, w412/w222} \draw (\u) -- (\v);
\end{tikzpicture}}
\caption{Stretching a vertex~$v$ with three children $c_1$, $c_2$, and $c_3$ gives this stretch gadget.}\label{fig:str}
\end{figure}

An anonymous reviewer suggested a simple alternative to stretching that,
in a network $N$,
below any vertex $v$ of out-degree $d \geq 1$ with children $c_1, c_2, \dots, c_d$,
inserts a recursively defined gadget $G_d$ instead
(again replacing the outgoing arcs of $v$):
If $d = 1$, then $G_d$ simply consists of $v$ and its child $c_1$,
with the arc in between exactly as it exists in $N$.
If $d \geq 2$, then $G_d$ is constructed by first inserting the gadget $G_{d - 1}$
between $v$ and children $c_1, c_2, \dots, c_{d - 1}$,
and then subdividing each arc $a$ of $G_{d - 1}$ with a new vertex $u_a$.
Finally, arcs from every $u_a$ to $c_d$ are added.
Note that, in the end, $v$ has a single outgoing arc.
To obtain a phylogenetic network, this arc needs to be contracted if $d \geq 2$,
but only once,
after the recursion is fully done.
This gadget mimics the property that removing a leaf from a tree with $k \geq 3$ leaves
produces what is essentially a subdivision of a tree with $k - 1$ leaves
(but its root may now have out-degree $1$).

These alternative gadgets can be described simply and concisely;
however, they can cause an increase in scanwidth that is exponential in the initial maximum out-degree of the network.
This is because, to add one more child, the number of arcs of the gadget is doubled through subdivision.
The next child to add then has high in-degree, which causes the gadget to have high scanwidth.
To obtain the running time reported in \cref{thm:stc},
we use stretch gadgets,
since these can only cause an increase in scanwidth that is linear in the initial maximum out-degree of the network
(\cref{lem:swstr}).

The proof of \cref{thm:stc} relies on four key claims:
\cref{lem:strsd,lem:binressd}, which state that stretching and in-splitting
preserve the collection of trees softly displayed by a network,
and \cref{lem:swstr,lem:swin}, which state that the increase in scanwidth when stretching and in-splitting is bounded,
using the network's maximum out-degree.
To prove \cref{lem:strsd}, we first formalize the intuition that the stretched network
``represents'' all possible binary out-resolutions.

\begin{lemma}\label{lem:strfdiffoutfd}
  Let $N$ be a network and let $T$ be a tree.
  Then,
  $\str(N)$ firmly displays $T$
  if and only if
  there is a binary out-resolution of $N$ that firmly displays $T$.
\end{lemma}
\begin{proof}
\proofright{}
Suppose $\str(N)$ firmly displays~$T$, that is,
there is some leaf-monotone subdigraph~$N'$ of $\str(N)$
that is isomorphic (respecting leaves)
to some subdivision~$T'$ of $T$,
and let~$\iota \colon V(T') \to V(N')$ denote this isomorphism.
Without loss of generality, we can assume that $\getroot(N') \in V(N)$
(so the root of $N'$ is not strictly inside a stretch gadget).
We note that $T$ is necessarily binary.
We then aim to establish the existence of a binary out-resolution of $N$ that firmly displays $T$.
To this end, we want to extract the structure of paths corresponding to $T$ in stretch gadgets.
This structure then allows us to construct a suitable binary out-resolution of $N$.
Informally, we first suppress all suppressible vertices of $N'$ that are strictly inside the stretch gadgets of $\str(N)$.
Now, there might be some vertices~$u$ of $N$ with out-degree $1$ left that have bifurcating vertices in their stretch gadgets,
and so we contract the top bifurcation onto~$u$.
More formally,
\begin{enumerate}
  \item we construct the set $S_1 \ceq \{v \in V(N') \setminus V(N) \mid \indeg_{N'} (v) = \outdeg_{N'} (v) = 1\}$,
  \item we obtain $N''$ from $N'$ by suppressing all vertices that are in $S_1$,
  \item we obtain $T''$ from $T'$ by suppressing all vertices that are in $\iota^{-1} (S_1)$,
  \item we construct the leaf-respecting isomorphism $\kappa \ceq \iota |_{V(T'')}$ between $T''$ and $N''$,
  \item we construct $S_2 \ceq \{u \in V(N'') \cap V(N) \mid \outdeg_{N''} (u) = 1 \land \outV_{N''} (u) \nsubseteq \outV_N (u)\}$,
  \item we obtain $N'''$ from $N''$ by, for each $u \in S_2$, contracting the child of $u$ into $u$, and
  \item we obtain $T'''$ from $T''$ by suppressing all vertices that are in $\kappa^{-1} (S_2)$.
\end{enumerate}
We argue that $T'''$ is a subdivision of $T$:
This is because $T'$ is a subdivision of $T$ by assumption
and we obtained $T'''$ from $T'$ by suppressing in-degree-$1$, out-degree-$1$ vertices,
which are not present in $T$ by definition of a phylogenetic tree.
We also argue that there is a leaf-respecting isomorphism between $T'''$ and $N'''$:
By assumption, $\iota$ is a leaf-respecting isomorphism between $T'$ and $N'$;
from these two, we constructed $T'''$ and $N'''$ while maintaining isomorphism (respecting leaves) throughout.
To finish this direction of the proof,
we show that there is a subdivision of $N'''$
that simultaneously is a leaf-monotone subdigraph of some binary out-resolution of $N$.
We can then conclude that such a binary out-resolution of $N$ firmly displays $T$.
To this end, we argue that any vertex $v \in V(N''')$, together with all its outgoing arcs in $N'''$,
also exists in a suitably chosen binary out-resolution of $N$,
possibly after performing subdivisions in $N'''$:
Note that $N'''$ is binary, and so $v$ has out-degree $0$, $1$, or $2$.
If $\outdeg_{N'''} (v) = 0$, then $v$ is a leaf in any binary out-resolution of $N$.
If $\outdeg_{N'''} (v) = 1$, then $v$ and its child $c \in \outV_{N'''} (v)$ both exist in $N$
and are connected by the arc $(v, c)$ in $N$ too,
and this arc also exists in a suitably chosen binary out-resolution of $N$.
If $\outdeg_{N'''} (v) = 2$, then we proceed as follows:
Either $v \in V(N)$ holds or $v$ originates from a stretch gadget.
We assume that $v \in V(N)$ holds
and handle all vertices originating from a possible stretch gadget involving $v$ together with $v$.
Directly below $v$ in $N'''$,
there is a binary tree (generally using vertices of a stretch gadget) that leads to a subset of $\outV_N (v)$.
Any binary tree leading to exactly $\outV_N (v)$ can be recreated in a binary out-resolution of $N$
through out-splitting below $v$.
To ensure that $N'''$ is a subdigraph of such a binary out-resolution,
we can subdivide an arc of the binary tree below $v$ in $N'''$
to account for vertices of $\outV_N (v)$ that do not appear in $N'''$
but do exist in any binary out-resolution of $N$ regardless.
In the binary out-resolution of choice,
these extra vertices can be placed in a binary tree that dangles off the vertex added through subdivision below $v$.

\proofleft{}
Suppose $T$ is firmly displayed by a binary out-resolution $N'$ of $N$.
To show that $\str(N)$ firmly displays $T$,
we argue that there is a subdivision of $N'$
that is isomorphic (respecting leaves) to a leaf-monotone subdigraph of $\str(N)$.
Concretely, for any vertex $v \in V(N')$,
there is an equivalent vertex in $\str(N)$,
possibly after subdividing arcs in $N'$:
Since $N'$ is a binary out-resolution,
$v$ has an out-degree of at most $2$ in $N'$.
If $\outdeg_{N'} (v) = 0$, then $v$ is a leaf of $N$ and, thus, a leaf of $\str(N)$.
If $\outdeg_{N'} (v) = 1$, then we know that $v$ is a reticulation and also present in $N$ and $\str(N)$,
always with the same neighborhood as in $N'$.
If $\outdeg_{N'} (v) = 2$, then we proceed as follows:
Note that either $v \in V(N)$ holds or $v$ was introduced in $N'$ through out-splitting.
We assume that $v \in V(N)$ holds
and handle all vertices introduced through out-splitting below $v$ in $N'$ together with $v$.
Rooted at $v$, these vertices form a binary tree whose leaves are $\outV_N (v)$
(although these are not necessarily leaves of $N'$ itself and $v$ is not necessarily the root of $N'$).
To finish the proof, we show that, for the vertices of this binary tree rooted at $v$ in $N'$,
there are equivalent vertices in $\str(N)$,
possibly after subdividing arcs of the binary tree in $N'$.
If $\outdeg_N (v) = 2$, then the binary tree has exactly two leaves and exists identically in $\str(N)$.
If $\outdeg_N (v) \geq 3$, then $v$ was stretched to obtain $\str(N)$,
and it remains to argue that a subdivision of the binary tree rooted at $v$ in $N'$
is isomorphic (respecting the leaves $\outV_N (v)$)
to a leaf-monotone subdigraph of the stretch gadget rooted at $v$ in $\str(N)$.
We use $B \subseteq N'$ to denote the binary tree rooted at $v$,
and we use $S \subseteq \str(N)$ to denote the stretch gadget rooted at $v$.
Note that $\getroot(B) = \getroot(S) = v$ and that $L(B) = L(S) = \outV_N (v)$.
What is left to show is equivalent to the claim that $S$ firmly displays $B$.
This follows from an argument that is analogous to that of Zhang \cite[Theorem 2]{Zha16}.
It is not essential that the stretch gadgets use different sorting networks than the universal networks by Zhang do.
\end{proof}

With \cref{lem:strfdiffoutfd} at hand,
we now show that stretching preserves the collection of trees that a network softly displays.

\begin{restatable}{lemma}{lemstrsd}\label{lem:strsd}
Let $N$ be any network and
let $T$ be any tree.
Then,
$N$ softly displays $T$
if and only if
$\str(N)$ softly displays $T$.
\end{restatable}
\begin{proof}\label{pr:lemstrsd}
\proofright{}
First, assume that $N$ softly displays $T$.
There is a binary out-resolution $N''$ of a binary in-resolution $N'$ of $N$
as well as a binary resolution $T'$ of $T$
such that $N''$ firmly displays $T'$.
From \cref{lem:strfdiffoutfd}, it follows that $\str(N')$ firmly displays $T'$.
Since $\str(N')$ is a binary resolution of $\str(N)$,
we conclude that $\str(N)$ softly displays $T$.

\proofleft{}
Second, suppose that $\str(N)$ softly displays $T$.
Then, there is a binary resolution $T'$ of $T$
that is firmly displayed by some binary resolution of $\str(N)$.
Any such binary resolution of $\str(N)$
may be assumed to have the form $\str(N')$
for some binary in-resolution $N'$ of $N$.
To summarize, $\str(N')$ firmly displays $T'$.
From \cref{lem:strfdiffoutfd}, it follows that
there is a binary out-resolution $N''$ of $N'$ that firmly displays $T'$.
Because $N''$ is a binary resolution of $N$,
this implies that $N$ softly displays $T$.
\end{proof}

We argue that a stretched network can be made binary through in-splitting
without altering the collection of softly displayed trees;
the collection of firmly displayed trees also stays unaltered.
In stretched networks, the maximum out-degree is $2$;
that is, we can assume that $\outDelta{N} = 2$.
The in-degrees, however, are generally unbounded.
The lemma below justifies arbitrary in-splitting (taking any binary (in-)resolution)
to make a stretched $N$ binary indeed.

\begin{restatable}{lemma}{lembinressd}\label{lem:binressd}
Let $N$ be any network such that $\outDelta{N} = 2$,
let $N'$ be any binary (in-)resolution of $N$,
and let $T$ be any tree.
Then,
\begin{enumerate}[(1)]
\item $N$ firmly displays $T$ if and only if $N'$ firmly displays $T$, and
\item $N$ softly displays $T$ if and only if $N'$ softly displays $T$.
\end{enumerate}
\end{restatable}
\begin{proof}\label{pr:lembinressd}
The ``firmly''-part (1) is a slight adaptation of a proposition by Gunawan~\cite[Proposition~3.9]{Gun18}.
As the proof is analogous, we omit it here and show only the ``softly''-part (2).

\proofright{}
Suppose that $N$ softly displays $T$.
Then, there are binary resolutions $N''$ of $N$ and $T'$ of $T$
such that $N''$ firmly displays $T'$.
Applying the ``firmly''-part (1) from right to left,
from $N''$ firmly displaying $T'$, it follows that $N$ firmly displays $T'$.
By applying the ``firmly''-part (1) from left to right,
from $N$ firmly displaying $T'$, we infer that $N'$ firmly displays $T'$.
Since $N'$ is its own binary resolution,
we can conclude that $N'$ softly displays $T$.

\proofleft{}
Suppose that $N'$ softly displays $T$.
Then, $N'$ firmly displays some binary resolution of $T$.
As $N'$ is a binary resolution of $N$,
we deduce that $N$ softly displays $T$.
\end{proof}

Finally, we show that neither stretching nor in-splitting to obtain a binary in-resolution
can blow up the scanwidth too much, in order to prove \cref{thm:stc}.

\begin{restatable}{lemma}{lemswstr}\label{lem:swstr}
  Let $N$ be any network.
  Then, $\sw(\str(N)) \leq \sw(N) + 2 \cdot \outDelta{N}$.
\end{restatable}
\begin{proof}\label{pr:lemswstr}
Let $\Gamma$ be any tree extension of $N$.
We prove the lemma in two steps.

First, let $v \in V(N)$ be a vertex with $d \ceq \outdeg_N (v) \geq 3$,
and let $N'$ be obtained from $N$ by stretching $v$,
thereby adding several new vertices (see \cref{def:strexp} in the appendix):
\begin{equation*}
\begin{split}
V(N') = V(N) &\uplus (\{u_{i, j} \mid i \in [2, d - 1] \land j \in [1, i]\} \\
&\uplus (\{u'_{i, j} \mid i \in [3, d - 1] \land j \in [2, i - 1]\} \\
&\uplus (\{u_{d, j} \mid j \in [2, d - 1]\} \\
&\uplus \{w_{i, j, k} \mid i, j \in [1, d - 1] \land k \in [1, 4]\}))) \eqdot
\end{split}
\end{equation*}
We construct a tree extension $\Gamma'$ of $N'$
by inserting the new vertices, $V(N') \setminus V(N)$,
in the form of a path right below $v$ in $\Gamma$,
above the children of $v$ in $\Gamma$.
The new vertices are inserted one after the other according to the following procedure:
\begin{enumerate}
\item For $i \gets 2$ to $d - 1$:
\begin{enumerate}
\item For $j \gets 1$ to $i$, insert $u_{i, j}$.
\item For $j \gets 2$ up to $i - 1$, insert $u'_{i, j}$.
\end{enumerate}
\item For $j \gets 2$ to $d - 1$, insert $u_{d, j}$.
\item For $i \gets 1$ to $d - 1$, for $j \gets 1$ to $d - 1$, for $k \gets 1$ to $4$, insert $w_{i, j, k}$.
\end{enumerate}
It is easy to verify that the resulting out-tree $\Gamma'$ is a tree extension of $N'$.

Second, we construct a tree extension $\Gamma''$ of $\str(N)$ as follows:
Starting from $\Gamma$,
we perform the above procedure for each vertex of $N$ that was stretched to obtain $\str(N)$,
inserting a new path each time.
We observe that $\sw(\Gamma'') \leq \sw(\Gamma) + 2 \cdot \outDelta{N}$ and,
since $\Gamma$ was chosen arbitrarily, $\sw(\str(N)) \leq \sw(N) + 2 \cdot \outDelta{N}$.
\end{proof}

The last ingredient for \cref{thm:stc} is the monotonicity of scanwidth under in-splitting.

\begin{restatable}{lemma}{lemswin}\label{lem:swin}
  Let $N'$ arise from any network $N$ through in-splitting.
  Then, $\sw(N') \leq \sw(N)$.
\end{restatable}
\begin{proof}\label{pr:lemswin}
Let $N$ be any arbitrary network,
let $v \in V(N)$ have in-degree at least $3$ in $N$,
let $u, w \in \inV_N (v)$ be two distinct parents of $v$ in $N$,
and let $N'$ be the result of in-splitting $u$ and $w$ from $v$ in $N$,
thereby adding a new vertex $x \in V(N') \setminus V(N)$.
Also, let $\Gamma$ be any tree extension of $N$.
Clearly, $v$ cannot be the root of $\Gamma$;
let $s \in \inV_\Gamma (v)$ be the parent of $v$ in $\Gamma$.
We construct $\Gamma'$ by subdividing the arc $(s, v) \in A(\Gamma)$ with $x$ in $\Gamma$.
It is easy to see that $\Gamma'$ is a tree extension of $N'$.
To conclude that $\sw(N') \leq \sw(N)$, as $\sw_{N'} (\Gamma') \leq \sw_N (\Gamma)$,
we make the following observations:
Let $t \in V(\Gamma') \setminus \{x, v\} \subseteq V(\Gamma)$ be a vertex.
\begin{enumerate}
\item We have
$\GW_t (\Gamma') \setminus \{(u, x), (w, x)\} \subseteq \GW_t (\Gamma) \setminus \{(u, v), (w, v)\}$.
\item If $(u, x) \in \GW_t (\Gamma')$, then $(u, v) \in \GW_t (\Gamma)$.
\item If $(w, x) \in \GW_t (\Gamma')$, then $(w, v) \in \GW_t (\Gamma)$.
\item We have
$\GW_x (\Gamma') \uplus \{(u, v), (w, v)\} \subseteq \GW_v (\Gamma) \uplus \{(u, x), (w, x)\}$.
\item We have
$\GW_v (\Gamma') \uplus \{(u, v), (w, v)\} \subseteq \GW_v (\Gamma) \uplus \{(x, v)\}$.
\end{enumerate}
Therefore, $\card{\GW_t (\Gamma')} \leq \card{\GW_t (\Gamma)}$,
$\card{\GW_x (\Gamma')} \leq \card{\GW_v (\Gamma)}$,
and $\card{\GW_v (\Gamma')} \leq \card{\GW_v (\Gamma)}$.
\end{proof}

We finish with the main result of this paper.

\begin{theorem}\label{thm:stc}
Given a network $N$, a tree $T$, and a tree extension $\Gamma$ of $N$,
we can decide whether $N$ softly displays $T$ in
$\bigO^* (2^{\bigO(\outDelta{T} \cdot k \cdot \log(k))})$ time,
where $k \ceq \sw_N (\Gamma) + \outDelta{N}$.
\end{theorem}
\begin{proof}
Let $N'$ be a binary (in-)resolution of $\str(N)$.
By \cref{lem:strsd,lem:binressd},
we know that $N$ softly displays $T$
if and only if $N'$ softly displays $T$.
Based on $\Gamma$,
we construct a tree extension $\Gamma'$ of $N'$
such that $\sw_{N'} (\Gamma')$ is bounded from above
by some linear function of $\sw_N (\Gamma) + \outDelta{N}$,
as described in the proofs of \cref{lem:swstr,lem:swin}.
It is easy to modify $N'$ and $\Gamma'$
so as to fulfill the requirements of \cref{as:NTGamma},
since we reject if $L(T) \nsubseteq L(N)$.
This involves removing excess leaves from $N'$,
updating $\Gamma'$ accordingly,
and making it canonical.
Finally, we invoke \cref{alg:stc},
from which we derive the promised running time;
see \cref{prop:stc}.
\end{proof}

\section{Conclusion}

We devised a parameterized algorithm for the \STC{} problem,
where it has to be decided whether a phylogenetic network softly displays a phylogenetic tree.
This algorithm consists of two steps:
solving the problem via dynamic programming when the network is binary
and reducing any non-binary network to an equivalent binary one.
The running time of this algorithm depends exponentially on a function of
the network's and the tree's maximum out-degree
and of the scanwidth attained by a given tree extension of the network.
The dependence on the input size is polynomial.

Future work can attempt to eliminate the superpolynomial dependence on the maximum out-degrees
or, alternatively, rule out the existence of such an algorithm.
Because of the dependence on scanwidth,
it is essential to be able to obtain low-scanwidth tree extensions in a reasonable amount of time.
This further motivates the search for exact or approximate algorithms to construct such tree extensions.
Additionally, it is of interest to tackle \STC{} with parameters stronger than scanwidth,
such as node-scanwidth~\cite{Hol23}, edge-treewidth~\cite{MPS+23}, or treewidth.
Even weaker parameters could be of interest,
when they give rise to algorithms that perform better in practice for example.
Finally, it is worth exploring whether our techniques can be adapted to approach other containment problems,
for instance, \NC{} in the sense of Janssen and Murakami \cite{JM21}.

\paragraph{Acknowledgments.}

We thank the reviewers and the editor for their valuable feedback and for the many improvements they suggested.
This includes the alternative gadgets mentioned in \cref{sec:arbnet},
which were suggested by an anonymous reviewer.

\bibliographystyle{abbrvnat}
\bibliography{mybibliography}

\clearpage

\appendix

\section{Stretching}

In the following, we list the vertices and arcs of a stretch gadget explicitly (compare with \cref{def:str}).
We hope that this can aid in possible implementation projects.

\begin{definition}\label{def:strexp}
Let $N$ be any arbitrary network,
and let $v \in V(N)$ be a vertex of out-degree $d \ceq \outdeg_N (v)$ at least $3$.
Then, \emph{stretching} $v$ produces the network $N' \ceq (V', A')$
by adding several new vertices, that is,
\begin{equation*}
\begin{split}
V' \ceq V(N) &\uplus (\{u_{i, j} \mid i \in [2, d - 1] \land j \in [1, i]\} \\
&\uplus (\{u'_{i, j} \mid i \in [3, d - 1] \land j \in [2, i - 1]\} \\
&\uplus (\{u_{d, j} \mid j \in [2, d - 1]\} \\
&\uplus \{w_{i, j, k} \mid i, j \in [1, d - 1] \land k \in [1, 4]\}))) \eqcom
\end{split}
\end{equation*}
while using the arcs
\begingroup
\allowdisplaybreaks
\begin{align*}
A' \ceq{}& (A(N) \setminus \outA_N (v)) \\
&\cup \{(v, u_{2, 1}), (v, u_{2, 2})\} \\
&\cup \{(u_{i, 1}, u_{i + 1, 1}), (u_{i, 1}, u_{i + 1, 2}) \mid i \in [2, d - 2]\} \\
&\cup \{(u_{i, i}, u_{i + 1, i}), (u_{i, i}, u_{i + 1, i + 1}) \mid i \in [2, d - 2]\} \\
&\cup \{(u_{i, j}, u'_{i, j}) \mid i \in [3, d - 1] \land j \in [2, i - 1]\} \\
&\cup \{(u'_{i, j}, u_{i + 1, j}), (u'_{i, j}, u_{i + 1, j + 1}) \mid i \in [3, d - 1] \land j \in [2, i - 1]\} \\
&\cup \{(u_{d - 1, 1}, w_{1, 1, 1}), (u_{d - 1, 1}, u_{d, 2})\} \\
&\cup \{(u_{d - 1, d - 1}, u_{d, d - 1}), (u_{d - 1, d - 1}, w_{1, d - 1, 2})\} \\
&\cup \{(u_{d, j}, w_{1, j - 1, 2}) \mid j \in [2, d - 1]\} \\
&\cup \{(w_{i, j, 1}, w_{i, j, 3}), (w_{i, j, 1}, w_{i, j, 4}) \mid i, j \in [1, d - 1]\} \\
&\cup \{(w_{i, j, 2}, w_{i, j, 3}), (w_{i, j, 2}, w_{i, j, 4}) \mid i, j \in [1, d - 1]\} \\
&\cup \{(w_{i, j, 4}, w_{i, j + 1, 1}) \mid i \in [1, d - 1] \land j \in [1, d - 2]\} \\
&\cup \{(w_{i, 1, 3}, w_{i + 1, 1, 1}) \mid i \in [1, d - 2]\} \\
&\cup \{(w_{i, d - 1, 4}, w_{i + 1, d - 1, 2}) \mid i \in [1, d - 2]\} \\
&\cup \{(w_{i, j, 3}, w_{i + 1, j - 1, 2}) \mid i \in [1, d - 2] \land j \in [2, d - 1]\} \\
&\cup \{(w_{d - 1, j, 3}, c_j) \mid j \in [1, d - 1]\} \\
&\cup \{(w_{d - 1, d - 1, 4}, c_d)\} \eqcom
\end{align*}
\endgroup
assuming that the children of $v$ in $N$ are $\outV_N (v) = \{c_1, c_2, \dots, c_d\}$.
Further, \emph{stretching} $N$ refers to stretching every vertex of $N$ that has an out-degree of at least~$3$,
and it produces the network $\str(N)$.
\end{definition}
\end{document}